\documentclass[a4paper,12pt]{amsart}

\usepackage{amsthm}
\usepackage{amsmath}
\usepackage{geometry}
\usepackage{amsfonts}
\usepackage{graphicx}
\usepackage{array}
\usepackage{amssymb}
\usepackage{mathtools}

	\renewcommand{\d}{\mathrm{d}}

\newtheorem{theorem}{Theorem}

\newtheorem{corollary}[theorem]{Corollary}
\newtheorem{proposition}[theorem]{Proposition}

\newtheorem{lemma}[theorem]{Lemma}

\numberwithin{equation}{section}
\numberwithin{theorem}{section}

\begin{document}

\title[Hamiltonian Field theories with boundaries]{Covariant Hamiltonian field theories on manifolds with boundary: Yang-Mills theories}

\author{A. Ibort}
\address{ICMAT and Depto. de Matem\'aticas, Univ. Carlos III de Madrid, Avda. de la
Universidad 30, 28911 Legan\'es, Madrid, Spain.}
\email{albertoi@math.uc3m.es}
\author{A. Spivak}
\address{Dept. of Mathematics, Univ. of California at Berkeley, 903 Evans Hall, 94720 Berkeley CA, USA}
\email{ameliaspivak@gmail.com}

\begin{abstract}
The multisymplectic formalism of field theories developed by many mathematicians over the last fifty years is extended in this work
to deal with manifolds that have boundaries. In particular, we develop a multisymplectic framework for first order covariant Hamiltonian field theories on manifolds with boundaries. This work is a geometric fulfillment of Fock's characterization of field theories as it appears in recent work by Cattaneo, Mnev and Reshetikhin \cite{Ca14}. This framework leads to a true geometric understanding of conventional choices for boundary conditions. For example, the boundary condition that the pull-back of the 1-form on the cotangent space of fields at the boundary vanish, i.e. $\Pi^*\alpha=0$ , is shown to be a consequence of our finding that the boundary fields of the theory lie in the 0-level set of the moment map of the gauge group of the theory. 

 It is also shown that the natural way to interpret Euler-Lagrange equations as an evolution system near the boundary is as a presymplectic system in an extended phase space containing the natural configuration and momenta fields at the boundary together with extra degrees of freedom corresponding to the transversal components at the boundary of the momenta fields of the theory. The consistency conditions at the boundary are analyzed and the reduced phase space of the system is determined to be a symplectic manifold with a distinguished isotropic submanifold corresponding to the boundary data of the solutions of Euler-Lagrange equations. This setting makes it possible to define well-posed boundary conditions, and provides the adequate setting for the canonical quantization of the system.  

The notions of the theory will be tested against three significant examples:  scalar fields, Poisson $\sigma$-model and Yang-Mills theories.
\end{abstract}

\maketitle
\tableofcontents


\section{Introduction}\label{sec:introduction}
Multisymplectic geometry provides a convenient framework for describing first order covariant field theories both in the Lagrangian and Hamiltonian formalism \cite{Ca91}. See also the GimMsy papers \cite{Go98} and \cite{Go04}. However, the role of boundaries in the multisymplectic formalism, relevant as it is in the construction of the corresponding quantum field theories, has not been incorporated in a unified geometrical picture of the theory.    
  
In this paper we will extend the multisymplectic formalism to deal with first-order covariant Hamiltonian field theories on manifolds with boundary, providing a consistent geometrical framework, for instance, for the perturbative quantization program recently set up by A. Cattaneo, P. Mnev and N. Reshetikhin \cite{Ca14} for theories on manifolds with boundary. The restriction to the boundary will provide the canonical Hamiltonian formalism needed for the canonical quantization picture that will be developed in detail elsewhere.  We will concentrate on the classical setting and we will prepare the ground to introduce a graded setting that will become useful when dealing with
the quantization of gauge theories. 

The history of the construction of a geometrical picture for field theories is extensive with many relevant contributions. We refer the reader to the comprehensive texts \cite{Gi09} and \cite{Bi11} and to references therein. The first author's earlier work on the subject, \cite{Ca91}, benefited from \cite{Ga72}, \cite{Go73}, \cite{Ki76}, \cite{Ki79} and many others. The ambitious GimMsy papers \cite{Go98} and \cite{Go04} were the first parts of a project that aimed to reconcile a multisymplectic geometrical formalism for field theories with the canonical picture needed for quantization. For recent work on the geometry of of first and higher order classical field theories see \cite{Gr12} and \cite{Gr15}. For recent work extending the multisymplectic formalism to higher order Hamiltonian theories see \cite{Le05}, \cite{Ec07}, \cite{Ro09}, \cite{Vi10}, \cite{Pr14} and references therein.

 In \cite{Ca11} the authors lay out Fock's unpublished account of the general structure of Lagrangian field theories on manifolds with boundary. Adjusting this account in the obvious
way, one arrives at the general structure of Hamiltonian field theories on manifolds with boundary $\mathit{a}$ $\mathit{la}$ Fock. The multisymplectic formalism we develop in this work to describe first-order Hamiltonian field theories on manifolds with boundary, puts meat and ligaments onto the bones, so to speak, of the Fock-inspired account of a Hamiltonian field theory. The immediate fruits of our formalism include a formula and a proof for the differential of the action functional, valid for any classical theory. In \cite{Ca11} and \cite{Ca14} and in many other works, for each classical field they consider, the authors have to come up with a different action functional and a different expression for the differential of the action functional. To do so they need to decide what the momenta of the theory need to be. In the multisymplectic formalism we develop, having one expression for the action functional and one expression for its differential that works for all classical theories, means in particular that we do not have to choose what the momenta fields should be for each physical theory we want to examine, our multisymplectic formalism identifies them for us. From there we are able to prove once and for all, that for all classical theories $\Pi(EL)$, the boundary values of the solutions the Euler-Lagrange equations, is an isotropic submanifold of $T^*\mathcal{F}_{\partial M}$ (see Sect. \ref{sec:general}). Our formalism is then shown to yield geometric insight into conventional choices for boundary conditions.

   This paper is devoted to systematically describing the classical ingredients in the proposed Hamiltonian framework. The description of the theory in the bulk, while  following along the lines already established in the literature, also includes analysis of the role of boundary terms in the computation of the critical points of the action functional.  Thus a natural relation emerges between the action functional and the canonical symplectic structure on the space of fields at the boundary.   It is precisely this relations that allows a better understanding of the role of boundary conditions. The canonical 1-form on the space of fields at the boundary can be directly related to the charges of the gauge symmetries of the theory, allowing us in this way to explain why admissible boundary conditions are determined by Lagrangian submanifolds on the space of fields at the boundary.

Once the geometrical analysis of the theory has been performed, the space of quantum states of the theory would be obtained, in the best possible situations,  by canonical or geometrical quantization of a reduced symplectic manifold of fields at the boundary that would describe its ``true'' degrees of freedom.   The propagator of the theory would be obtained by quantizing a Lagrangian submanifold of the reduced phase space of the theory provided by the specification of admissible boundary conditions. The latter should preserve the fundamental symmetries of the theory, in the sense that the charges associated to them should be preserved. The resulting overall picture as descrbed in the case of Chern-Simons theory  \cite{At90}, is that the functor defining a quantum field theory is obtained by geometric quantization of the quasi-category of Lagrangian submanifolds associated to admissible boundary conditions at the boundaries of space-times and their corresponding fields. (This picture is being currently extended to the Poisson $\sigma$-model \cite{Co13}, \cite{CC14}.)

The level of rigor of this work is that of standard differential geometry: When dealing with finite-dimensional objects, they will be smooth differentiable manifolds, locally trivial bundles, etc., however when dealing with infinite-dimensional spaces, we will assume, as customary, that the rules of global differential calculus apply and we will use them freely without providing constructions that will lead to bona fide Banach manifolds of maps and sections.  Also, the notation of variational differentials and derivatives will be used for clarity without attempting to discuss the classes of spaces of generalised functions needed to justify their use.

The paper is organized as follows:  Section \ref{sec:general} is devoted to summarizing the basic geometrical notions underlying the theory. The multisymplectic formalism is briefly reviewed, the action principle and a fundamental formula exhibiting the differential of the action functional of the theory is presented and proved. The role of symmetries, moment maps at the boundary and boundary conditions are elucidated.   Section \ref{sec:presymplectic} presents the evolution formulation of the theory near the boundary.  The presymplectic picture of the system will be established and the subsequent constraints analysis is laid out.    Its relation with reduction with respect to the moment map at the boundary is pointed out.   Real scalar fields and the Poisson $\sigma$-model are analyzed to illustrate the theory.    Finally, Section \ref{sec:Yang-Mills}  concentrates on the study of Yang-Mills theories on manifolds with boundary as first-order Hamiltonian field theories in the multisymplectic framework and the Hamiltonian reduced phase space of the theory is described.



\section{The multisymplectic formalism for first order covariant Hamiltonian field theories on manifolds with boundary}\label{sec:general}


\subsection{The setting: the multisymplectic formalism}\label{sec:multisymplectic}

The geometry of Lagrangian and Hamiltonian field theories has been examined in the literature from varying perspectives. For our purposes here we single out for summary the Hamiltonian multisymplectic description of field theories on manifolds without boundary found in \cite{Ca91}. Everything in this section will apply also to manifolds possessing boundaries. In the next section we will consider only  manifolds having boundaries and we will extend the multisymplectic formalism to deal with Hamiltonian field theories over such manifolds.

A manifold $M$ will model the space or spacetime at each point of which the classical field under discussion assumes a value. We will therefore take $M$ to be an oriented $m = 1+d$ dimensional smooth manifold.   In most situations $M$ is either Riemannian or Lorentzian and time-oriented. We will denote the metric on $M$ by $\eta$. In either case we will denote by $\mathrm{vol}_M$ the volume form defined by the metric $\eta$ on $M$. In an arbitrary local chart $x^{\mu}$ this volume form takes the form $\mathrm{vol}_M = \sqrt{|\eta|} \d x^0 \wedge \d x^1 \cdots \wedge \d x^d $.  Notice however that the only structure on $M$ required to provide the kinematical setting of the theory will be a volume form $\mathrm{vol}_M$ and, unless specified otherwise, local coordinates will be chosen such that $\mathrm{vol}_M = \d x^0 \wedge \d x^1 \cdots \wedge \d x^d$.

The fundamental geometrical structure of a given theory will be provided by a fiber bundle over $M$, $\pi \colon E \to M$. Local coordinates adapted to the fibration will be denoted as $(x^\mu, u^a)$, $a= 1, \ldots, r$, where $r$ is the dimension of the standard fiber.   

Let $J^1E$ denote the first jet bundle of the bundle $E$, i.e.,  at each point $(x,u) \in E$, the fiber of $J^1E$ consists of the set of equivalence classes of germs of sections of $\pi\colon E \rightarrow M$. If we let $\pi_1^0$ be the projection map, $\pi_1^0 \colon J^1E \rightarrow E$, then $(J^1E,\pi_1^0,E)$ is an affine bundle over $E$ modelled on the linear bundle $VE \otimes \pi^*(T^*M)$ over $E$. (See \cite{Sa89}, \cite{Ca91} and \cite{Gr15} for details on affine geometry and the construction of the various affine bundles naturally associated to $E \to M$.)

If $(x^{\mu};u^a)$, is a bundle chart for the bundle $\pi \colon E \to M$ then we will denote by $(x^{\mu},u^a;u_{\mu}^a)$ a local chart for $J^1E$.

The affine dual of $J^1E$ is the vector bundle over $E$ whose fiber at $\xi = (x,u)$ is the linear space of affine maps $\mathrm{Aff}(J^1E_\xi, \mathbb{R})$.   The vector bundle $\mathrm{Aff}(J^1E, \mathbb{R})$ possesses a natural subbundle defined by constant functions along the fibers of $J^1E \to E$, that we will denote again, abusing notation, as $\mathbb{R}$.  The quotient bundle $\mathrm{Aff}(J^1E, \mathbb{R})/\mathbb{R}$ will be called the covariant phase space bundle of the theory, or the phase space for short.   Notice that such bundle, denoted in what follows by $P(E)$ is the vector bundle with fiber at $\xi = (x,u) \in E$ given by $(V_uE\otimes T_x^*M)^* \cong T_xM \otimes(V_uE)^*\cong \mathrm{Lin}(V_uE,T_xM)$ and projection $\tau_1^0 \colon P(E) \to E$.

Local coordinates on $P(E)$ can be introduced as follows:
Affine maps on the fibers of $J^1E$ have the form $u_{\mu}^a \mapsto \rho_0 + \rho_a^{\mu}u_{\mu}^a$ where $u_{\mu}^a$ are natural coordinates on the fiber over the point $\xi$ in $E$ with coordinates $(x^{\mu},u^a)$. Thus an affine map on each fiber over $E$ has coordinates $\rho_0, \rho^{\mu}_a$, with $\rho^\mu_a$ denoting linear coordinates on $TM \otimes VE^*$ associated to bundle coordinates $(x^\mu, u^a)$.   Functions constant along the fibers are described by the numbers $\rho_0$, hence elements in the fiber of $P(E)$ have coordinates $\rho_a^{\mu}$.  Thus a bundle chart for the bundle $\tau_1^0\colon P(E) \to E$ is given by $(x^\mu, u^a;  \rho^\mu_a)$.

The choice of a distinguished volume form $\mathrm{vol}_M$ in $M$ allows us to identify the fibers of $P(E)$ with a subspace of $m$-forms on $E$ as follows (\cite{Ca91}):
The map $u_{\mu}^a \mapsto \rho_a^{\mu}u_{\mu}^a$ corresponds to the $m$-form $\, \, \rho_a^\mu \, \d u^a\wedge \mathrm{vol}_{\mu}$
where vol$_{\mu}$ stands for $i_{{\partial}/{\partial x^{\mu}}}$vol$_M.$
Let ${\bigwedge}^m (E)$ denote the bundle of $m$-forms on $E$. Let 
${\bigwedge}_k^m(E)$ be the subbundle of ${\bigwedge}^m (E)$ consisting of those $m$-forms which vanish when $k$ of their arguments are vertical.  So in our local coordinates, elements of ${\bigwedge}_1^m(E)$, i.e., $m$-forms on $E$ that vanish when one of their arguments is vertical, commonly called semi-basic 1-forms, have the form $\rho_a^{\mu} \, \d u^a\wedge \mathrm{vol}_\mu + \rho_0 \mathrm{vol}_M$, and elements of ${\bigwedge}_0^m(E)$, i.e., basic $m$-forms, have the form $\rho_0\mathrm{vol}_M$. These bundles form a short exact sequence:
$$
0\rightarrow\textstyle{\bigwedge}^m_0E\hookrightarrow
	\textstyle{\bigwedge}^m_1E\rightarrow P(E)\rightarrow0 \, .
$$
Hence ${\bigwedge}_1^m E$ is a real line bundle over $P(E)$ and, for each point $\zeta = (x,u,\rho)\in P(E)$, the fiber is the quotient $\bigwedge_1^m (E) _\zeta/ \bigwedge_0^m (E)_\zeta$.
		
  The bundle $\textstyle{\bigwedge}_1^m(E)$ carries a canonical $m$--form  which may be defined by a generalization of the definition of the canonical 1-form on the cotangent bundle of a manifold.  Let $\sigma \colon \textstyle{\bigwedge}_1^m(E) \to E$  be the canonical projection, then the canonical $m$-form $\Theta$ is defined by 
$$
\Theta_\varpi(U_1,U_2,\ldots,U_m) = \varpi(\sigma_*U_1, \ldots, \sigma_*U_m)
$$
where $\varpi\in\bigwedge^m_1(E)$ and $U_i\in T_\varpi(\bigwedge^m_1(E))$.
As described above, given bundle coordinates $(x^\mu,u^a)$ for $E$
we have coordinates $(x^\mu,u^a,\rho, \rho^\mu_a)$
on $\bigwedge^m_1(E)$ adapted to them  
and the point $\varpi\in\bigwedge^m_1(E)$ with coordinates
$(x^\mu,u^a;\rho, \rho^\mu_a )$ is the $m$-covector
$\varpi =  \rho^\mu_a\, \d u^a\wedge \mathrm{vol}_\mu + \rho \, \mathrm{vol}_M$.
With respect to these same coordinates we have the local expression
$$
\Theta =  \rho^\mu_a \, \d u^a \wedge \mathrm{vol}_\mu  + \rho\, \mathrm{vol}_M \, ,
$$
for $\Theta$, where $\rho$ and $\rho^\mu_a$ are now to be interpreted as coordinate
functions.
	
The $(m+1)$-form $\Omega = \d \Theta $ defines a multisymplectic structure on the manifold $\bigwedge_1^m(E)$, i.e.$(\bigwedge^m_1(E),\Omega)$ is a multisymplectic manifold. There is some variation in the literature on the definition of multisymplectic manifold. For us, following \cite{Ca91}, \cite{Go98} and \cite{Ca99}, a multisymplectic manifold is a pair $(X,\Omega)$ where $X$ is a manifold of some dimension $m$ and $\Omega$ is a $d$-form on $X$, $d \geq 2$, and $\Omega$ is closed and nondegenerate. By nondegenerate we mean that if $i_v{\Omega} = 0$ then $v=0$.
	
We will refer to $\bigwedge^m_1E$, by $M(E)$ to emphasize that it is a multisymplectic manifold. We will denote the projection $M(E) \to E$
by $\nu$, while the projection $M(E) \to P(E)$ will be
denoted by $\mu$.  Thus $\nu = \tau^0_1\circ\mu$, with $\tau_1^0 \colon P(E) \to E$
the canonical projection.(See figure 1.)

A Hamiltonian $H$ on $P(E)$ is a section of
$\mu$. Thus in local coordinates 
$$
H(\rho_a^{\mu}\, \d u^a\wedge\mathrm{vol}_{\mu}) = \rho_a^{\mu}\d u^a\wedge \mathrm{vol}_{\mu}-\mathbf{H}(x^{\mu},u^a, \rho_a^{\mu}) \mathrm{vol}_M \, ,
$$ 
where $\mathbf{H}$ is here a real-valued function also called the Hamiltonian function of the theory.
	
We can use the Hamiltonian section $H$ to define an $m$-form on $P(E)$
by pulling back the canonical $m$-form $\Theta$ from $M(E)$.  We call
the form so obtained the Hamiltonian $m$-form associated with $H$ and denote
it by $\Theta_H$. Thus if we write the section defined in local coordinates $(x^\mu, u^a;\rho, \rho_a^\nu )$ as 
\begin{equation}\label{rhoH}
\rho = - \mathbf{H}(x^{\mu}, u^a, \rho_a^\mu ) \, ,
\end{equation}
then
\begin{equation}\label{ThetaH}
	\Theta_H = \rho_a^\mu\, \d u^a \wedge \mathrm{vol}_\mu - \mathbf{H}(x^\mu,u^a, \rho_a^\mu) \, \mathrm{vol}_M \, .
\end{equation}
\,
 In Eqs. \eqref{rhoH} and \eqref{ThetaH}, the minus sign in front of the Hamiltonian is chosen to be in keeping with the traditional conventions in mechanics for the integrand of the action over the manifold. When the form $\Theta_H$ is pulled back to the manifold $M$, as described in Section \ref{sec:sections}, the integrand of the action over $M$ will have a form reminiscent of that of mechanics, with a minus sign in front of the Hamiltonian. See equation \eqref{action_phip}.  In what follows, unless there is risk of confussion, we will use the same notation $H$ both for the section and the real-valued function $\mathbf{H}$ defined by a Hamiltonian.

	
\subsection{The action and the variational principle}\label{sec:variational}


\subsubsection{Sections and fields over manifolds with boundary}\label{sec:sections}
From here on, in addition to being an oriented smooth manifold with either a Riemannian or a Lorentzian metric, $M$ has a boundary $\partial M$. The orientation chosen on $\partial M$ is consistent with the orientation on $M$. Everything in the last section applies. The presence of boundaries, apart from being a natural ingredient in any attempt of constructing a field theory, will enable us to enlarge the use to which the multisymplectic formalism can be applied, starting with the statement and proof of Lemma 2.1.	

The fields $\chi$ of the theory in the Hamiltonian formalism constitute a class of sections of the bundle $\tau_1 : P(E) \to M$.    $P(E)$ is a bundle over $E$ with projection $\tau_1^0$ and it is a bundle over $M$ with projection $\tau_1 = \pi\circ \tau_1^0$.  The sections that will be used to describe the classical fields in the Hamiltonian formalism are those sections $\chi\colon M \to P(E)$, i.e.  $\tau_1\circ \chi  = \mathrm{id}_M$,  such that $\chi = P \circ \Phi$ where $\Phi \colon M \rightarrow E$ is a section of $\pi: E\rightarrow M$, i.e. $\pi \circ \Phi = \mathrm{id}_M$, and  $P \colon E \to P(E)$ is a section of $\tau_1^0 \colon P(E) \rightarrow E$ i.e. $\tau_1^0 \circ P = \mathrm{id}_P$. (See Figure \ref{sections}).  The sections $\Phi$ will be called the configuration fields or just the configurations, and the sections $P$ the momenta fields of the theory.   In other words $u^a = \Phi^a(x)$ and $\rho_a^\mu = P_a^\mu (\Phi(x))$ will provide local expression for the section $\chi = P \circ \Phi$.
 We will denote such a section $\chi$ by $(\Phi, P)$ to stress the iterated bundle structure  of $P(E)$ and we will refer to $\chi$ as a double section.

\begin{figure}[ht]
\centering
\includegraphics[width=8cm]{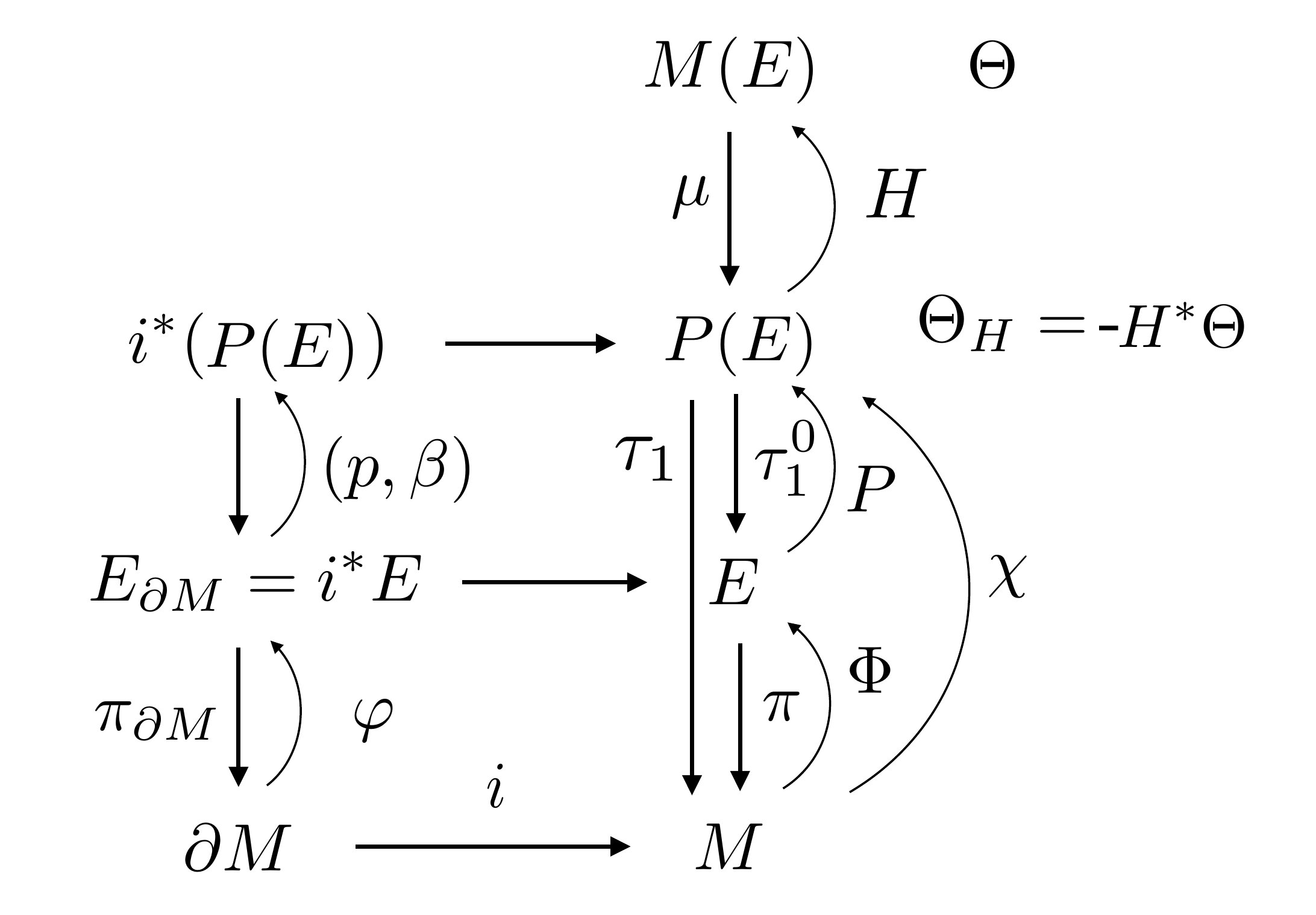}
\caption{Bundles, sections and fields: configurations and momenta}\label{sections}
\end{figure}

We will denote by $\mathcal{F}_M$ the space of sections $\Phi$ of the bundle $\pi \colon E \rightarrow M $, that is $\Phi \in \mathcal{F}_M$, and we will denote by $\mathcal{F}_{P(E)}$ the space of double sections $\chi = (\Phi, P)$.  Thus $\mathcal{F}_{P(E)}$  represents the space of fields of the theory, configurations and momenta, in the first order covariant Hamiltonian formalism. 

The equations of motion of the theory will be defined by means of a variational principle, i.e., they will be related to the critical points of an action functional $S$ on $\mathcal{F}_{P(E)}$. Such action will be given simply by
\begin{equation}\label{action}
 S(\chi ) = \int_M \chi^*\Theta_H \, ,
 \end{equation}
or in a more explicit notation,
\begin{equation}\label{action_phip}
S(\Phi,P) = \int_M \left( P_a^\mu (x) \partial_\mu \Phi^a (x) - H(x,\Phi(x), P(x)) \right) \mathrm{vol}_M ,
\end{equation}
where $P_a^\mu (x)$ is shorthand for $P_a^\mu (\Phi (x))$.  Of course, as is usual in the derivations of equations of motion via variational principles, we assume that the integral in Eq. \eqref{action} is well defined.  It is also assumed that the `differential' symbol in equation $(2.5)$ below, defined in terms of directional derivatives, is well defined and that the same is true for any other similar integrals that will appear in this work.

\begin{lemma}\label{dS }
 With the above notations we obtain,
\begin{equation}\label{dSfirst}
\mathrm{d} S (\chi) (U) = \int_M \chi^* \left(i_{\widetilde U} \d \Theta _H \right) + \int_{\partial M} (\chi\circ i)^* \left(i_{\widetilde U} \Theta_H \right) \, , 
\end{equation}
where $U$ is a vector field on $P(E)$ along the section $\chi$, $\widetilde{U}$ is any extension of $U$ to a tubular neighborhood of the range of $\chi$, and $i\colon \partial M \to M$ is the canonical embedding.
\end{lemma}

\begin{proof}  If $\chi$ is a section of $P(E)$, then we denote by $T_\chi \mathcal{F}_{P(E)}$ the tangent bundle to the space of fields at $\chi$. Tangent vectors $U$ to $\mathcal{F}_{P(E)}$ at $\chi$, i.e., $U \in T_\chi \mathcal{F}_{P(E)}$, are just vector fields $U$ on $P(E)$ along the map $\chi$ or, in other words, maps $U \colon M \to TP(E)$ such that $\tau_{P(E)}\circ U = \chi$, where $\tau_{P(E)} \colon TP(E) \to P(E)$ denotes the canonical tangent bundle projection.

Thus if $U\in T_\chi \mathcal{F}_{P(E)}$, with $U(x) \in T_{\chi (x)}P(E)$, then 
consider a curve $\chi_\lambda (x) = \chi (\lambda, x) \colon (-\epsilon, \epsilon) \times M \to P(E)$, such that $\chi (0, x) = \chi (x)$, and 
$$
U(\chi(x)) =\left. \frac{\partial}{\partial \lambda}\right|_{\lambda = 0} \chi(\lambda, x) \, .
$$
We can extend the vector field $U$ to a tubular neighborhood $T_\chi$ of the image of $\chi$ in $P(E)$ and we will denote it by $\tilde{U}$.  Consider the local flow $\varphi_\lambda$ of $\tilde{U}$,
$$
\frac{\d}{\d\lambda} \varphi_\lambda = \tilde{U} \circ \varphi_\lambda \, ,
$$
or in other words, let us denote the integral curves of $\tilde{U}$ by $\varphi_\lambda (\xi)$,  $\xi \in T_\chi \subset P(E)$.
Then if $\xi = \chi (x)$ we have, $\varphi_\lambda(\xi) = \varphi_\lambda (\chi (x))= \chi (\lambda, x) = (\chi\circ \chi_\lambda)(x)$, i.e., $\varphi_\lambda \circ \chi = \chi_\lambda$. We thus obtain,
\begin{eqnarray}
\d S(\chi) (U)  &=& \left.\frac{\d}{\d\lambda}\right|_{\lambda = 0} S(\chi_\lambda)  =  \left. \frac{\d}{\d\lambda}\right|_{\lambda = 0} \int_M \chi_\lambda^* \Theta_H = \nonumber \\
&=&  \int_M \left. \chi^*\frac{\partial}{\partial\lambda}\right|_{\lambda = 0} \varphi_\lambda^*\Theta_H = \int_M \chi^* (\mathcal{L}_{\tilde{U}} \Theta_H) = \nonumber  \\ 
&=& \int_M \chi^* d(i_{\tilde{U}}\Theta_H) + \int_M \chi^* i_{\tilde{U}} \d \Theta _H \, . \label{3rdline}
\end{eqnarray}

Applying Stokes' theorem to the first term in eq. \eqref{3rdline} then yields eq. \eqref{dSfirst}.
\end{proof}


\subsubsection{The cotangent bundle of fields at the boundary}\label{sec:cotangent_boundary}

The boundary term contribution to $\mathrm{d} S$ in eq. \eqref{dSfirst}, that is, $\int_{\partial M} (\chi\circ i)^* \left(i_{\tilde U} \Theta_H\right)$, suggests that there is a family of fields at the boundary that play a special role.  Actually, we notice that the field $\tilde{U}$ being vertical with respect to the projection $\tau_1\colon P(E) \to M$ has the local form $\tilde{U} = A^a \,  \partial/\partial u^a + B_\mu^a  \, \partial/\partial \rho_\mu^a$. Hence we obtain for the boundary term,
\begin{equation}\label{boundaryfirst}
\int_{\partial M} (\chi\circ i)^* \left(i_{\widetilde U} \Theta_H \right) = \int_{\partial M} (\chi\circ i)^*  \rho_a^\mu \, A^a  \, \mathrm{vol}_\mu  = \int_{\partial M} i^*(P_a^\mu\, A^a \, \mathrm{vol}_\mu)
\end{equation}
for $\chi = (\Phi, P)$. 

We will assume now and in what follows, that there exists a collar around the boundary $U_\epsilon \cong (-\epsilon, 0] \times \partial M$.  We choose local coordinates $(x^0,x^k)$, on the collar such that $x^0 = t \in (-\epsilon, 0]$ , and $x^k$, $k = 1, \ldots, d$, define local coordinates for $\partial M$. In these coordinates $\mathrm{vol}_{U_\epsilon} = \d t \wedge \mathrm{vol}_{\partial M}$ with $\mathrm{vol}_{\partial M}$ a volume form on $\partial M$. The r.h.s. of eq. \eqref{boundaryfirst} becomes,
\begin{equation}\label{boundary_final}
\int_{\partial M} i^*(P_a^\mu\, A^a \, \mathrm{vol}_\mu)  = \int_{\partial M} p_a \, A^a \, \mathrm{vol}_{\partial M} \, ,
\end{equation}
where $p_a = P_a^0\circ i$ is the restriction to $\partial M$ of the zeroth component of the momenta field $P_a^\mu$ in a local coordinate chart of the previous form.

 Consider the space of fields at the boundary obtained by restricting the zeroth component of sections $\chi$ to $\partial M$, that is the fields of the form (see Figure \ref{sections})
$$
\varphi^a = \Phi^a \circ i \, , \qquad p_a = P_a^0 \circ i \, .
$$
Notice that the fields $\varphi^a$ are nothing but sections of the bundle $i^*E$, the pull-back along $i$ of the bundle $E$, while the space of fields $p_a$ can be thought of as 1-semibasic $d$-forms on $i^*E \to \partial M$.   This statement is made precise in the following:

\begin{lemma}\label{decomposition}   Given a collar around $\partial M$, $U_\epsilon \cong (-\epsilon, 0] \times \partial M$, and a volume form $\mathrm{vol}_{\partial M}$ on $\partial M$ such that $\mathrm{vol}_{U_{\epsilon}} = \d t \wedge \mathrm{vol}_{\partial M}$ with $t$ the normal coordinate in $U_\epsilon$, then the pull-back bundle $i^*(P(E))$ is a bundle over the pull-back bundle $i^*E$ and decomposes 
naturally as $i^*P(E) \cong \bigwedge_1^m(i^*E) \oplus \bigwedge_1^{m-1}(i^*E)$.   If $i^*\zeta \in i^*P(E)$, we will denote by $p$ and $\beta$ the components of the previous decomposition, that is, $i^*\zeta = p + \beta$.
\end{lemma}

\begin{proof}  By definition of pull-back, the fiber over a point $x \in \partial M$ of the bundle $i^*E$, consists of all vectors in $E_x$.  The pull-back bundle $i^*P(E)$ is a bundle over $i^*E$, the fiber over $(x,u) \in i^*E$ is $T_xM\otimes VE_u^*$.   Using the volume form $\mathrm{vol}_M$, we identify this fiber with $\bigwedge^{m-1}(T_xM)\otimes VE_u^*$ by contracting elements $\varpi = v \otimes \alpha \in T_xM\otimes VE_u^*$ with $\mathrm{vol}_M(x)$.  The collar neighborhood $U_\epsilon$ introduces a normal coordinate $t\in (-\epsilon, 0]$ such that $\mathrm{vol}_{U_{\epsilon}}
 = \d t \wedge \mathrm{vol}_{\partial M}$.  Notice that such decomposition depends on the choice of the collar. We obtain $\varpi = \rho_a^0 \d u^a \wedge \mathrm{vol}_{\partial M} + \rho_a^k \d u^a \wedge \d t \wedge i_{\partial /\partial x^k}\mathrm{vol}_{\partial M}$.  Finally, the assignment $\varpi \mapsto (\rho_a^0 \d u^a \wedge \mathrm{vol}_{\partial M} , \rho_a^k \d u^a \wedge \wedge i_{\partial /\partial x^k}\mathrm{vol}_{\partial M} )$ provides the decomposition we are after and $p_a = \rho_a^0$, $\beta_a^k = \rho_a^k$.
\end{proof}

If we denote by $\mathcal{F}_{\partial M}$ the space of configurations of the theory, $\varphi^a$, i.e., $\mathcal{F}_{\partial M} = \Gamma(i^*E)$, then the space of momenta of the theory $p_a$ can be identified with the space of sections of the bundle $\bigwedge_1^m(i^*E) \to i^*E$, according to Lemma \ref{decomposition}.   Therefore the space of fields $(\varphi^a, p_a)$ can be identified with the contangent bundle $T^*\mathcal{F}_{\partial M}$ over $\mathcal{F}_{\partial M}$ in a natural way, i.e., each field $p_a$ can be considered as the covector at $\varphi^a$ that maps the tangent vector $\delta\varphi$ to $\mathcal{F}_{\partial M}$ at $\varphi$ into the number $\langle p, \delta \varphi \rangle$ given by,
\begin{equation}\label{pairing_cotangent}
\langle p, \delta \varphi \rangle = \int_{\partial M} p_a(x)\delta\varphi^a (x) \, \mathrm{vol}_{\partial M} \, .
\end{equation}
	
Notice that the tangent vector $\delta \varphi$ at $\varphi$ is a vertical vector field on $i^*E$ along $\varphi$, and the section $p$ is a 1-semibasic $m$-form on $i^*E$ (Lemma \ref{decomposition}). Hence the contraction of $p$ with $\delta\varphi$ is an $(m-1)$-form along $\varphi$, and its pull-back $\varphi^*\langle p, \delta\varphi \rangle$ along $\varphi$ is an $(m-1)$-form on $\partial M$ whose integral defines the pairing above, Eq. \eqref{pairing_cotangent}. 	
	
 Viewing the cotangent bundle $T^*\mathcal{F}_{\partial M}$ as double sections $(\varphi, p)$ of the bundle $\bigwedge_1^m(i^*E) \to i^*E \to \partial M$ described by Lemma \ref{decomposition}, the canonical 1-form $\alpha$ on $T^*\mathcal{F}_{\partial M}$ can be expressed as,
\begin{equation}\label{alpha}
\alpha_{(\varphi, p)} (U) = \int_{\partial M} p_a (x) \delta\varphi^a (x) \, \mathrm{vol}_{\partial M}
\end{equation}
where $U$ is a tangent vector to $T^*\mathcal{F}_{\partial M}$ at $(\varphi, p)$, that is, a vector field on the space of 1-semibasic forms on $i^*E$ along the section $(\varphi^a, p_a)$, and therefore of the form $U = \delta\varphi^a \, \partial /\partial u^a + \delta p_a \, \partial /\partial \rho_a$.
	
Finally, notice that the pull-back to the boundary map $i^*$, defines a natural map from the space of fields in the bulk, $\mathcal{F}_{P(E)}$, into the phase space of fields at the boundary $T^*\mathcal{F}_{\partial M}$.  Such map will be denoted by $\Pi$ in what follows, that is, 
$$
\Pi \colon \mathcal{F}_{P(E)}\to T^*\mathcal{F}_{\partial M} \, , \qquad \Pi(\Phi, P) = (\varphi, p) , \, \quad  \varphi = \Phi\circ i, \, p_a = P_a^0\circ i \, .
$$

With the notations above, by comparing the expression for the boundary term given by eq. \eqref{boundary_final}, and the expression for the canonical 1-form $\alpha$, eq. \eqref{alpha}, we obtain,
$$
\int_{\partial M} (\chi\circ i)^* \left(i_{\tilde U} \Theta_H\right) = (\Pi^*\alpha)_\chi (U) \, .
$$
In words, the boundary term in eq. \eqref{dSfirst} is just the pull-back of the canonical 1-form $\alpha$ at the boundary along the projection map $\Pi$.

In what follows it will be customary to use the variational derivative notation when dealing with spaces of fields. For instance, if $F(\varphi,p)$ is a differentiable function defined on $T^*\mathcal{F}_{\partial M}$ we will denote by $\delta F / \delta \varphi^a$ and $\delta F / \delta p_a$ functions (provided that they exist) such that
\begin{equation}\label{dF}
\d F_{(\varphi,p)}(\delta \varphi^a, \delta p_a) = \int_{\partial M} \left( \frac{\delta F}{\delta \varphi^a} \delta \varphi^a + \frac{\delta F}{\delta p_a} \delta p_a \right) \mathrm{vol}_{\partial M} \, ,
\end{equation}
with $U = (\delta \varphi^a, \delta p_a)$ a tangent vector at $(\varphi,p)$. 
We also use an extended Einstein's summation convention such that integral signs will be omitted when dealing with variational differentials. For instance,
\begin{equation}\delta F =   \frac{\delta F}{\delta \varphi^a} \delta \varphi^a + \frac{\delta F}{\delta p_a} \delta p_a \, ,
\end{equation} 
will be the notation that will replace $\d F$ in Eq. \eqref{dF}.   Also in this vein we will write,
$$
\alpha = p_a \, \delta \varphi^a \, ,
$$
and the canonical symplectic structure $\omega_{\partial M} = -\d \alpha$ on $T^*\mathcal{F}_{\partial M}$ will be written as,
$$
\omega_{\partial M} = \delta \varphi^a \wedge \delta p_a \, ,
$$
by which we mean
$$
\omega_{\partial M} ((\delta_1\varphi^a, \delta_1p_a), (\delta_2\varphi^a, \delta_2p_a)) = \int_{\partial M} \left( 
\delta_1\varphi^a(x) \delta_2 p_a(x) - \delta_2\varphi^a (x) \delta_1p_a(x) \right) \mathrm{vol}_{\partial M} \, ,
$$
where $(\delta_1\varphi^a, \delta_1p_a), (\delta_2\varphi^a, \delta_2p_a)$ are two tangent vectors at $(\varphi, p)$.


\subsubsection{Euler-Lagrange's equations and Hamilton's equations}
We now examine the contribution from the first term in $\d S$, eq. \eqref{dSfirst}.
Notice that such a term can be 
thought of as a 1-form on the space of fields on the bulk, $\mathcal{F}_{P(E)}$.  We will call it the Euler-Lagrange 1-form and denote it by $\mathrm{EL}$, thus with the notation of Lemma \ref{dS },
$$
\mathrm{EL}_\chi (U) = \int_M \chi^* \left(i_{\tilde U} \d \Theta _H \right)  \, .
$$
A double section $\chi = (\Phi, P)$ of $P(E) \to E \to M$ will be said to satisfy the Euler-Lagrange equations determined by the first-order Hamiltonian field theory defined by $H$, if $\mathrm{EL}_\chi = 0$, that is, if $\chi$ is a
zero of the Euler-Lagrange 1-form $\mathrm{EL}$ on $\mathcal{F}_{P(E)}$.   Notice that this is equivalent to 
\begin{equation}\label{formEL}
	\chi^*(i_{\tilde{U}} \\d \Theta _H)=0 \, ,
\end{equation}
for all vector fields $\tilde{U}$ on a tubular neighborhood of the image of $\chi$ in $P(E)$. 
The set of all such solutions of Euler-Lagrange equations will be denoted by $\mathcal{EL}_M$ or just $\mathcal{EL}$ for short.\
 
 In local coordinates $x^\mu$ such that the volume element takes the form $\mathrm{vol_M}= dx^0 \wedge\cdots \wedge dx^d$, and for natural local coordinates $(x^\mu,u^a,\rho^\mu_a)$ on $P(E)$, using eqs. \eqref{rhoH}, \eqref{ThetaH}, we have,
\begin{eqnarray*}
	i_{\partial/\partial \rho^\mu_a}\d \Theta _H &=& -\frac{\partial
	 H}{\partial \rho^\mu_a} d^mx + \d u^a\wedge \d^{m-1}x_\mu  \\
	i_{\partial/\partial u^a} \d \Theta _H &=& -\frac{\partial
	 H}{\partial u^a} d^mx - d\rho^\mu_a\wedge \d^{m-1}x_\mu.
\end{eqnarray*}
Applying Eq. \eqref{formEL} to these last two equations we obtain the Hamilton equations for the field in the bulk:

\begin{equation}\label{hamilton_equations}
	\frac{\partial u^a}{\partial x^\mu} = \frac{\partial H}{\partial \rho^\mu_a}\, ; \qquad
	\frac{\partial \rho^\mu_a}{\partial x^\mu} = -\frac{\partial H}{\partial u^a} \, ,
\end{equation}
where a summation on $\mu$ is understood in the last equation.   Note that had we not changed to normal coordinates on $M$, the volume form would not have the above simple form and therefore there would be related extra terms in the previous expressions and in Eqs. \eqref{hamilton_equations}.

These Hamilton equations are often described as being covariant. This term
must be treated with caution in this context. Clearly, by writing the equations
in the invariant form $\chi^*(i_{\tilde{U}}\d \Theta _H)=0$ we have shown that they are in
a sense covariant. However, it is important to remember that the function $H$
is, in general, only locally defined; in other words, there is in general no true
`Hamiltonian function', and the local representative $H$ transforms in a
non-trivial way under coordinate transformations. When $M(E)$ is a trivial
bundle over $P(E)$, so that there is a predetermined global section, then
the Hamiltonian section may be represented by a global function and no problem
arises. This occurs for instance when $E$ is trivial over $M$.
In general, however, there is no preferred section of $M(E)$ over
$P(E)$ to relate the Hamiltonian section to, and in order to write the
Hamilton equations in manifestly covariant form one must introduce a
connection. (See \cite{Ca91} for a more detailed discussion and \cite{Gr12} for a general treatment of these issues.)


\subsection{The fundamental formula}\label{sec:fundamental}

Thus we have obtained the formula that relates the differential of the action with a 1-form on a space of fields on the bulk manifold and a 1-form on a space of fields at the boundary.
\begin{equation}\label{fundamental}
\mathrm{d} S_\chi = \mathrm{EL}_\chi +  \Pi^* \alpha_\chi \, , \qquad \chi \in \mathcal{F}_{P(E)} \, . 
\end{equation}
In the previous equation $\mathrm{EL}_\chi$ denotes the Euler-Lagrange 1-form on the space of fields $\chi = (\Phi, P)$ with local expression (using variational derivatives):
\begin{equation}\label{ELform}
\mathrm{EL}_\chi = \left(  \frac{\partial \Phi^a}{\partial x^\mu}  - \frac{\partial H}{\partial P^\mu_a}  \right) \delta P_a^\mu  - \left( 	\frac{\partial P^\mu_a}{\partial x^\mu} + \frac{\partial H}{\partial \Phi^a}  \right) \delta \Phi^a \, ,
\end{equation}
or, more explicitly:
$$
\mathrm{EL}_\chi (\delta \Phi, \delta P) = \int_M \left[ \left(  \frac{\partial \Phi^a}{\partial x^\mu}  - \frac{\partial H}{\partial P^\mu_a}  \right) \delta P_a^\mu  - \left( \frac{\partial P^\mu_a}{\partial x^\mu} + \frac{\partial H}{\partial \Phi^a}  \right) \delta \Phi^a \right] \, \mathrm{vol}_M \, .
$$

In what follows we will denote by $(P(E), \Theta_H)$ the covariant Hamiltonian field theory with bundle structure $\pi \colon E \to M$ defined over the $m$-dimensional manifold with boundary $M$, Hamiltonian function $H$ and canonical $m$-form $\Theta_H$.    

We will say that the action $S$ is regular if the set of solutions of Euler-Lagrange equations $\mathcal{EL}_M$ is a submanifold of $\mathcal{F}_{P(E)}$. 
Thus we will also assume when needed that the action $S$ is regular (even though this must be proved case by case) and that the projection $\Pi(\mathcal{EL})$ to the space of fields at the boundary $T^*\mathcal{F}_{\partial M}$ is a smooth manifold too.

 This has the immediate implication that the projection of $\mathcal{EL}$ to the boundary $\partial M$ is an isotropic submanifold:

\begin{proposition}\label{isotropicEL}  Let $(P(E), \Theta_H)$ be a first order Hamiltonian field theory on the manifold $M$ with boundary, with regular action $S$ and such that $\Pi(\mathcal{EL})\subset T^*\mathcal{F}_{\partial M}$ is a smooth submanifold.    Then $\Pi (\mathcal{EL}) \subset T^*\mathcal{F}_{\partial M}$ is an isotropic submanifold.
\end{proposition}

\begin{proof}   Along the submanifold $\mathcal{EL} \subset T^*\mathcal{F}_{\partial M}$ we have,
$$
\d S \mid_{\mathcal{EL}} = \Pi^*\alpha\mid_{\mathcal{EL}} .
$$
Therefore $\mathrm{d} (\Pi^*\alpha) = \mathrm{d}^2 S = 0$ along $\mathcal{EL}$, and $\mathrm{d} (\Pi^*\alpha) = \Pi^* \mathrm{d} \alpha$ along $\mathcal{EL}.$
But $\Pi$ being a submersion then implies that $\mathrm{d}\alpha = 0$ along $\Pi(\mathcal{EL})$. 
\end{proof}

In many cases $\Pi(\mathcal{EL})$ is not only isotropic but Lagrangian. We will come back to the analysis of this in later sections.


\subsection{Symmetries and the algebra of currents}

Without attempting a comprehensive description of the theory of symmetry for covariant Hamiltonian
field theories, we will describe some basic elements needed in what follows (see details in \cite{Ca91}). 
Recall from Sect. \ref{sec:multisymplectic}, $(M(E),\Omega)$ is a multisymplectic manifold with $(m+1)$-dimensional multisymplectic form $\Omega = d \Theta $,  where dim $M = m$. Canonical transformations in the multisymplectic framework for Hamiltonian field theories are diffeomorphisms $\Psi \colon M(E) \to M(E)$ such that $\Psi^*\Omega = \Omega$.   Notice that if $\Psi$ is a diffeomorphism such that $\Psi^*\Theta = \Theta$, then $\Psi$ is a canonical transformation.       

A distinguished class of canonical transformations is provided by those transformations $\Psi$ induced by diffeomorphisms $\psi_E \colon E \to E$, i.e., $\Psi (\varpi) = (\psi_E^{-1})^*\varpi$, $\varpi \in M(E)$.   If the diffeomorphism $\psi_E$ is a bundle isomorphism, there will exist another diffeomorphism $\psi_M \colon M \to M$ such that $\pi\circ \psi_E = \psi_M \circ \pi$.  Under such circumstances it is clear that the induced map $ (\psi_E^{-1})^* \colon \bigwedge^m(E) \to \bigwedge^m(E)$ preserves both $\bigwedge_1^m(E)$ and $\bigwedge_0^1(E)$, thus the map $\Psi = (\psi_E^{-1})^*\colon M(E) \to M(E)$ induces a natural map $\psi_* \colon P(E) \to P(E)$ such that $\mu \circ (\psi_E^{-1})^* = \psi_* \circ \mu$.   Canonical transformations induced from bundle isomorphisms will be called covariant canonical maps.  

Given a one-parameter group of canonical transformations $\Psi_t$, its infinitesimal generator $U$ satisfies
$$
\mathcal{L}_U \Omega = 0 \, .
$$
Vector fields $U$ on $M(E)$ satisfying the previous condition will be called (locally) Hamiltonian vector fields.   Locally Hamiltonian vector fields $U$ for which there exists a $(m-1)$-form $f$ on $M(E)$ (we are assuming that $\Omega$ is a $(m+1)$-form) such that
$$
i_U \Omega = \d f \ ,
$$
will be called, in analogy with mechanical systems, (globally) Hamiltonian vector fields.  The class $\mathbf{f} = \{ f + \beta \mid \d\beta = 0, \, \beta \in \Omega^{m-1}(M(E)) \}$ determined by the $(m-1)$-form $f$ is called the Hamiltonian form of the vector field $U$ and such a vector field will be denoted as $U_{\mathbf{f}}$.   

The Lie bracket of vector fields induces a  Lie algebra structure on the space of Hamiltonian vector fields that we denote as $\mathrm{Ham}(M(E),\Omega)$.  Notice that Hamiltonian vector fields whose flows $\Psi_t$ are defined by covariant canonical transformations are globally Hamiltonian because $\mathcal{L}_U \Theta = 0$, and therefore
$i_U \d \Theta = - \d i_U\Theta$. The Hamiltonian form associated to $U$ is the class containing the $(m-1)$-form $f = i_U\Theta$.

The space of Hamiltonian forms, denoted in what follows by $\mathcal{H}(M(E))$, carries a canonical bracket defined by
$$
\{ \mathbf{f}, \mathbf{f}' \} = i_{U_f} i_{U_{f'}} \Omega + Z^{m-1}(M(E)) \, ,
$$
where $Z^{m-1}(M(E))$ denotes the space of closed $(m-1)$-forms on $M(E)$.  
The various spaces introduced so far are related by the short exact sequence \cite{Ca91}:
$$
0 \to H^{m-1} (M(E))  \to \mathcal{H}(M(E)) \to \mathrm{Ham}(M(E), \Omega) \to 0 \, .
$$

Let $G$ be a Lie group acting on $E$ by bundle isomorphisms and $\psi_g \colon E \to E$, the diffeomorphism defined by the group element $g \in G$. This action induces an action on the multisymplectic manifold $(M(E), \Omega)$ by canonical transformations.   
Given an element $\xi \in \mathfrak{g}$, where now and in what follows $\mathfrak{g}$ denotes the Lie algebra of the Lie group $G$, we will denote by $\xi_{M(E)}$ and $\xi_E$ the corresponding vector fields defined by the previous actions on $M(E)$ and $E$, respectively.      The vector fields $\xi_{M(E)}$ are Hamiltonian with Hamiltonian forms $\mathbf{J}_\xi$, that is,
\begin{equation}\label{ixiJxi}
i_{\xi_{M(E)}} \Omega = \d J_\xi \, ,
\end{equation}
with $J_\xi = i_{\xi_{M(E)}}\Theta$.
It is easy to check that
$$
\{Ê\mathbf{J}_\xi,  \mathbf{J}_\zeta \} = \mathbf{J}_{[\xi , \zeta]}  + c(\xi, \zeta) \, 
$$
where $c(\xi, \zeta)$ is a cohomology class of order $m-1$.  The bilinear map $c(\cdot, \cdot)$ defines an element in $H^2(\mathfrak{g}, H^{m-1}(M(E)))$ (see \cite{Ca91}).  In what follows we will assume that the group action is such that the cohomology class $c$ vanishes. Such actions are called strongly Hamiltonian (or just Hamiltonian, for short).

So far our discussion has not involved a particular theory, that is, a Hamiltonian $H$.     Let $(P(E), \Theta_H)$ be a covariant Hamiltonian field theory and $G$ a Lie group acting on $\mathcal{F}_{P(E)}$.   
Among all possible actions of groups on the space of double sections $\mathcal{F}_{P(E)}$ those that arise from an action on $P(E)$ by covariant canonical transformations are of particular significance.   Let $G$ be a group acting on $E$ by bundle isomorphisms. Let $\psi_*(g)$ denote the covariant diffeomorphism on $P(E)$ defined by the group element $g$. Then the transformed section $\chi^g$ is given by $\chi^g (x) = \psi_*(g)(\chi(\psi_M(g^{-1}) x))$ where $\psi_M(g)$ is the diffeomorphism on $M$ defined by the action of the group. We will often consider only bundle automorphisms over the identity, in which case $\chi^g (x) = \psi_*(g)(\chi (x))$.  Such bundle isomorphisms will be called gauge transformations and the corresponding group of all gauge transformations will be called the gauge group of the theory and denoted by $\mathcal{G}(E)$, or just $\mathcal{G}$ for short, in what follows.  

The group $G$ will be said to be a symmetry of the theory if  $S(\chi^g) = S(\chi)$ for all $\chi \in \mathcal{F}_{P(E)}$, $g\in G$.
Notice that, in general, an action of $G$ on $M(E)$ by bundle isomorphisms will leave $\Theta$ invariant and will pass to the quotient space $P(E)$, however it doesn't have to preserve $\Theta_H$.
Hence, it is obvious that a group $G$ acting on $P(E)$ by covariant transformations will be a symmetry group of the Hamiltonian field theory defined by $H$ iff $g^*\Theta_H = \Theta_H + \beta_g$, where now, for the ease of notation, we indicate the diffeomorphism $\psi_*(g)$ simply by $g$, and $\beta_g$ is a closed $m$-form on $M$.      In what follows we will assume that the group $G$ acts on $E$ and its induced action on $P(E)$ preserves the $m$-form $\Theta_H$, that is $\beta_g = 0$ for all $g$.

Because the action of the group $G$ preserves the $m$-form $\Theta_H$, the group acts by canonical transformation on the manifold $(P(E), d \Theta_H)$ with Hamiltonian forms $\mathbf{J}_\xi$ given by (the equivalence class determined by the $m$-forms):
$$
J_\xi = i_{\xi_{P(E)}}\Theta_H \, .
$$

\begin{theorem}[Noether's theorem]\label{Noether}
Let $G$ be a Lie group acting on $E$ which is a symmetry group of the Hamiltonian field theory $(P(E), \Theta_H)$ and such that it preserves the $m$-form $\Theta_H$.  If $\chi \in \mathcal{EL}$ is a solution of the Euler-Lagrange equations of the theory, then the $(m-1)$-form $\chi^*J_{\xi}$ on $M$ is closed.
\end{theorem}

\begin{proof}   Because $\chi$ is a solution of Euler-Lagrange equations, recalling eq. \eqref{formEL} we have
$$
0 = \chi^*(i_{\xi_{P(E)}}\Omega_H) = \chi^*\d J_\xi = \d(\chi^*J_\xi) \, .
$$ 
\end{proof}

The de Rham cohomology classes determined by the closed $(m-1)$-forms $\chi^*J_{\xi}$ on $M$ will be called currents and denoted by $\mathbf{J}_\xi[\chi]$.  
Using the Poisson bracket $\{ \cdot, \cdot \}$ defined on the space of Hamiltonian forms $\mathcal{H}(P(E))$ we define a Lie bracket in the space of currents $\mathbf{J}_\xi[\chi] \in H^{m-1}(M)$ by
$$
\{  \mathbf{J}_\xi[\chi], \mathbf{J}_\zeta[\chi] \} = \chi^* \{  \mathbf{J}_\xi, \mathbf{J}_\zeta \} = \mathbf{J}_{[\xi, \zeta]}[\chi] \, .
$$ 
By Stokes' theorem, the $(m-1)$-forms $i^*(\mathbf{J}_\xi[\chi])$ on $\partial M$ satisfy
\begin{equation}\label{boundary_conservation}
\int_{\partial M} i^*\mathbf{J}_\xi[\chi] = 0 \, .
\end{equation}
We will refer to the quantity $Q \colon \mathcal{F}_{P(E)} \to \mathfrak{g}^*$, where $\mathfrak{g}^*$ denotes the dual of the Lie algebra $\mathfrak{g}$, defined by 
\begin{equation}\label{chargeQ}
\langle Q(\chi), \xi \rangle = \int_{\partial M} i^* \mathbf{J}_\xi[\chi]  \, , \qquad  \forall \xi \in \mathfrak{g}  \, ,
\end{equation}
as the charge defined by the symmetry group.  Notice that the pairing $\langle \cdot, \cdot \rangle$ on the left hand side of Eq. \eqref{chargeQ} is the natural pairing between $\mathfrak{g}$ and $\mathfrak{g}^*$.   As a consequence of Noether's theorem we get $Q\mid_{\mathcal{EL}} = 0$.

\subsubsection{The moment map at the boundary}

Suppose that there is an action of a Lie group $G$ on the bundle $E$ that leaves invariant the restriction of the bundle $E$ to the boundary, that is, the transformations $\Psi_g$ defined by the elements of the group $g\in G$ restrict to the bundles $i^*(P(E))$ and $E_{\partial M} := i^*E$ (see Figure \ref{sections}).   We will denote such restriction as $\Psi_g\mid_{\partial M} = g_{\partial M}$.
 
Two elements $g, g' \in G$ will induce the same transformation on the bundle $E_{\partial M}$ if there exists an element $h$ such that $g' = gh$ and $h_{\partial M} = \mathrm{id}_{E_{\partial M}}$.  If we consider now the group $\mathcal{G}$ of all gauge transformations, then the set of group elements that restrict to the identity at the boundary is a normal subgroup of $\mathcal{G}$ which we will denote by $\mathcal{G}_0$.   The induced action of $\mathcal{G}$ at the boundary is the 
 action of the group $\mathcal {G}_{\partial M} = \mathcal{G}/\mathcal{G}_0$ which is the group of gauge transformations of the bundle $E_{\partial M} = i^*E$. 

In particular the group $\mathcal{G}$ induces an action on $\mathcal{F}_{\partial M}$ by
$$
g\cdot \varphi (x) = \psi_E(g)(\Phi (g^{-1}x)) = g_{\partial M} (\varphi (g^{-1} x)) \, , \qquad \forall x \in \partial M,\,  g \in \mathcal{G}\, ,
$$
and similarly for the momenta field $p$.    

\begin{proposition}\label{boundary_moment}   Let $\mathcal{G}_{\partial M}$ denote the gauge group at the boundary, that is, the group whose elements are the transformations induced at the boundary by gauge transformations of $E$. Then the action of $\mathcal{G}_{\partial M}$ in the space of fields at the boundary is strongly Hamiltonian with moment map $$
\mathcal{J} \colon T^*\mathcal{F}_{\partial M} \to \mathfrak{g}_{\partial M}^*
$$ 
given by,
$$
\langle \mathcal{J}(\varphi, p), \xi \rangle = \langle Q(\chi), \xi \rangle \, \qquad \forall \xi \in \mathfrak{g}_{\partial M} \, ,
$$
where $\Pi(\chi) = (\varphi,p)$, and $\mathfrak{g}_{\partial M}$, $\mathfrak{g}^*_{\partial M}$ denote respectively the Lie algebra and the dual of the Lie algebra of the group $\mathcal{G}_{\partial M}$.
In other words, the projection map $\Pi$   composed with the moment map at the boundary $\mathcal{J}$ is the charge $Q$ of the symmetry group.
\end{proposition}

\begin{proof}  The action of the group $\mathcal{G}_{\partial M}$ on $T^*\mathcal{F}_{\partial M}$ is by cotangent liftings, thus its moment map $\mathcal{J}$ takes the particularly simple form,
$$
\langle \mathcal{J}(\varphi,p) , \xi \rangle = \langle p, \xi_{\mathcal{F}_{\partial M}}(\varphi) \rangle \, ,
$$
where $\xi_{\mathcal{F}_{\partial M}}$ denotes, consistently, the infinitesimal generator defined by the action of $\mathcal{G}_{\partial M}$ on $\mathcal{F}_{\partial M}$.   Such generator, because the action is by gauge transformations, i.e., bundle isomorphisms over the identity, has the explicit expression:
$$
\xi_{\mathcal{F}_{\partial M}} = \xi \circ \varphi \frac{\delta}{\delta \varphi} \, ,
$$ 
where $\xi$ is the infinitesimal generator of the action of $\mathcal{G}_{\partial M}$ on $E_{\partial M}$.   Notice that the Lie algebra ·$\mathfrak{g}$ of the group of gauge transformations is precisely the algebra of vertical vector fields on $E$ (similarly for $\mathfrak{g}_{\partial M}$ and $E_{\partial M}$).   Hence $\xi \in \mathfrak{g}_{\partial M}$ is just a vertical vector field and the infinitesimal generator $\xi_{\mathcal{F}_{\partial M}}$ at $\varphi$, which is just a tangent vector to $\mathcal{F}_{\partial M}$ at $\varphi$, is the vector field along $\varphi$ given by the composition $\xi \circ \varphi$. 

However because the action of $\mathcal{G}_{\partial M}$ on $E_{\partial M}$ is exactly the action of $\mathcal{G}$ on $E\mid_{\partial M}$,  $\xi$ can be considered as an element on $\mathfrak{g}$ and (recalling the definition of the pairing in $T^*\mathcal{F}_{\partial M}$, eq. \eqref{pairing_cotangent}, and the discussion at the end of Sect. \ref{sec:cotangent_boundary} on the conventions with variational derivatives) we get,
\begin{eqnarray}
\langle p, \xi_{\mathcal{F}_{\partial M}}(\varphi) \rangle &=& \int_{\partial M} p_a\,  \xi^a (\varphi (x))\,  \mathrm{vol}_{\partial M}  \label{momentum_Q1} \\ &=& \int_{\partial M} i^*\left(\chi^* \left(i_{\xi_{P(E)}}\Theta_H \right) \right) = \int_{\partial M} i^*\mathbf{J}_\xi [\chi] = \langle Q(\chi), \xi \rangle \, , \label{momentum_Q2}
\end{eqnarray}
with $\Pi (\chi) = (\varphi, p)$.   In the previous computations we have used that $i_{\xi_{P(E)}}\Theta_H = \rho_a^\mu \,\xi^a (x,\rho^a) \, \mathrm{vol}_\mu$, therefore $\chi^*\left(i_{\xi_{P(E)}}\Theta_H \right)= P_a(\Phi(x))^\mu \,  \xi^a (\Phi(x)) \,  \mathrm{vol}_\mu$ and thus
$$
i^*\left(\chi^*\left(i_{\xi_{P(E)}}\Theta_H \right) \right)= p^a(\varphi(x)) \, \xi^a(\varphi(x)) \, \mathrm{vol}_{\partial M} \, .
$$
Notice the particularly simple form that the currents take in this situation $J_\xi[\chi] = p_a \xi^a(\varphi)$.
\end{proof}

Thus Noether's Theorem (which implies that $Q|_{\mathcal{EL}} = 0$) together with Prop. \ref{boundary_moment}, imply that for any $\chi = (\Phi,P) \in \mathcal{EL}$, then $(\varphi,p) \in \mathcal{J}^{-1}(\mathbf{0})$, with $(\varphi, p) = \Pi (\chi)$.\\

The main, and arguably the most significant, example of symmetries is provided by theories such that the symmetry group is the full group of automorphisms of the bundle $\pi \colon E \to M$, and in particular its normal subgroup of bundle automorphisms over the identity map, i.e., diffeomorphisms $\psi_E\colon E \to E$ preserving the structure of the bundle and such that $\pi \circ \psi_E = \psi_E$.   As indicated before, such group will be called the gauge group of the theory (or, better the group of gauge transformations of the theory) and we will denote it by $\mathcal{G}(E)$ (or just $\mathcal{G}$ if there is no risk of confusion).     In such case, eqs. \eqref{momentum_Q1}-\eqref{momentum_Q2}, imply the following:

\begin{corollary}\label{Qalpha} With the above notations, for the group of gauge transformations we obtain,
$$
Q = \Pi^*\alpha \, ,
$$
  where $\alpha$ is the canonical 1-form on $T^*\mathcal{F}_{\partial M}$,
in the sense that for any $\xi\in \mathfrak{g}$ and $\chi\in J^1\mathcal{F}^*$, 
$$
\langle Q(\chi), \xi \rangle = \alpha_{\Pi(\chi)}(\xi_{\mathcal{F}_{\partial M}}) = \Pi^*\alpha_{\chi}(\xi_{P(E)}) \, .
$$
\end{corollary}


\subsection{Boundary conditions}  Because of the boundary term $\Pi^*\alpha$ arising in the computation of the critical points of the action $S$,  the propagator of the corresponding quantum theory will be affected by such contributions and the theory could fail to be unitary \cite{As15}.   One way to avoid this problem is by selecting a subspace of the space of fields $\mathcal{F}_{P(E)}$ such that $\Pi^*\alpha$ will vanish identically when restricted to it (see for instance an analysis of this situation in Quantum Mechanics in \cite{As05}.)

Moreover, we would like to choose a maximal subspace with this property.  Then these two requirements will amount to choosing boundary conditions determined by a maximal submanifold $\mathcal{L} \subset T^*\mathcal{F}_{\partial M}$ such that $\alpha\mid_{\mathcal{L}} = 0$, that is, $\mathcal{L}$ is a special Lagrangian submanifold of the cotangent bundle $T^*\mathcal{F}_{\partial M}$.   

In general we will consider not just a single boundary condition but a family of them defining a Lagrangian fibration of $T^*\mathcal{F}_{\partial M}$.  An example of such a choice is the Lagrangian fibration $\mathcal{L}$ corresponding to the vertical fibration of $T^*\mathcal{F}_{\partial M}$.  For $\varphi \in \mathcal{F}_{\partial M}$, the subspace of fields defined by the leaf $\mathcal{L}_{\varphi}$, $\varphi \in \mathcal{F}_{\partial M}$ is just the subspace of fields $\chi = (\Phi, P)$ such that $\Phi\mid_{\partial M} = \varphi$.   

Another argument justifying the use of special Lagrangian submanifolds of $T^*\mathcal{F}_{\partial M}$ as boundary conditions, relies just on the structure of the classical theory and its symmetries and not on its eventual quantization.  Recall from Cor. \ref{Qalpha}, that if a theory $(P(E), \Theta_H)$ has the group of gauge transformations $\mathcal{G}$ of the bundle $E$ as a symmetry group, then $Q = \Pi^*\alpha$. Therefore the admissible fields of the theory - not necessarily solutions of Euler-Lagrange equations - are those such that the charge $Q$ is preserved along the boundary, that is, those that lie on a special Lagrangian submanifold $\mathcal{L} \subset T^*\mathcal{F}_{\partial M}$.

We will say that a (classical field) theory is Dirichlet if, for any $\varphi \in \mathcal{F}_{\partial M}$, there exists a unique solution $\chi = (\Phi, P)$ of the Euler-Lagrange equations, i.e., $\chi \in \mathcal{EL}$ such that $\Phi |_{\partial M} = \varphi$. 

\begin{theorem}\label{Lagrangian_sub}
Let $S$ be a regular action defined by a first order Hamiltonian field theory $(P(E), \Theta_H)$ then, if the theory is Dirichlet, $\Pi(\mathcal{EL})$ is a Lagrangian submanifold of $T^*\mathcal{F}_{\partial M}$.
\end{theorem}
\begin{proof}  Recall the discussion in Sec. \ref{sec:fundamental}. If the action is regular, i.e. if the solutions of the Euler-Lagrange equations $\mathcal{EL}$ define a submanifold of $\mathcal{F}_{P(E)}$, then from the fundamental relation eq. \eqref{fundamental}, we get Prop. \ref{isotropicEL}, that $\Pi(\mathcal{EL})$ is an isotropic submanifold of $T^*\mathcal{F}_{\partial M}$.

 Let the functional $W$ denote the composition $W = S\circ D$ where $S : \mathcal{F}_{P(E)} \rightarrow \mathbb{R}$ is the action of our theory defined by $H$, i.e. $S(\chi)= \int_M \chi^*\Theta_H \,$ and $D : \mathcal{F}_{\partial M} \rightarrow \mathcal{F}_{P(E)}$ is the map that assigns to any boundary data $\varphi$ the unique solution $(\Phi, P)$  of the Euler-Lagrange equations such that $\Phi |_{\partial M} =\varphi$.
Thus 
$$
W(\varphi)= \int_M \chi^*\Theta_H  = S (\chi) \, ,
$$ 
for any $\varphi \in \mathcal{F}_{\partial M}$.  By Eq. \eqref{ELform} and since $D(\phi)= (\Phi,P)=\chi \in \mathcal{EL}$ it follows that $\mathrm{EL}_{\chi}=0$.  A simple computation then shows that
$$
\mathrm{d}W (\varphi) (\delta \varphi) = p_a\delta \varphi^a \, ,
$$
where $p_a = P_a^0\mid_{\partial M}$.
Thus the graph of the 1-form $\mathrm{d} W =\Pi(\mathcal{EL})$, i.e., $W$ is a generating function for $\Pi(\mathcal{EL})$ and therefore $\Pi(\mathcal{EL})$ is a Lagrangian submanifold of $\mathcal{T^*F}_{\partial M}$.

\end{proof}

The Dirichlet condition can be weakened and a corresponding proof for a natural extension of Theorem \ref{Lagrangian_sub}  can be provided. See our work on Hamiltonian dynamics \cite{Ib16}.

\newpage

\section{The presymplectic formalism at the boundary}\label{sec:presymplectic}


\subsection{The evolution picture near the boundary}\label{sec:dynamical_eqs}

  We discuss in what follows the evolution picture of the system near the boundary.  As discussed in Section \ref{sec:cotangent_boundary}, we assume that there exists a collar $U_\epsilon \cong (-\epsilon , 0]\times\partial M$ of the boundary \ $\partial M$ with adapted coordinates $(t; x^1,\ldots, x^d)$, where $t = x^0$ and where $x^i$, $i = 1,\ldots, x^d$ define a local chart in $\partial M$. The normal coordinate $t$ can be used as an evolution parameter in the collar. We assume again that the volume form in the collar is of the form $\mathrm{vol}_{U_{\epsilon}} = dt\wedge \mathrm{vol}_{\partial M}$. 
  
 If $M$ happens to be a globally hyperbolic space-time $M \cong [t_0,t_1]\times \Sigma$ where $\Sigma$ is a Cauchy surface, $[t_0,t_1] \subset \mathbb{R}$ denotes a finite interval in the real line, and the metric has the form $-dt^2 + g_{\partial M}$ where $g_{\partial M}$ is a fixed Riemannian metric on $\partial M$, then $t$ represents a time evolution parameter throughout the manifold and the volume element has the form $\mathrm{vol}_M = dt \wedge \mathrm{vol}_{\partial M}$. Here, however, all we need to assume is that our manifold has a collar at the boundary as described above.

Restricting the action $S$ of the theory to fields defined on $U_\epsilon$, i.e., sections of the pull-back of the bundles $E$ and $P(E)$ to $U_\epsilon$, we obtain,
\begin{equation}\label{Sepsilon}
S_\epsilon (\chi ) = \int_{U_\epsilon} \chi^*\Theta_H = \int_{-\epsilon}^0 \d t \int_{\partial M} \mathrm{vol}_{\partial M} \left[ P_a^0 \partial_0 \Phi^a + P_a^k\partial_k \Phi^a - H(\Phi^a, P_a^0, P_a^k) \right]\, . 
\end{equation}
Defining the fields at the boundary as discussed in Lemma \ref{decomposition},
$$
\varphi^a = \Phi^a|_{\partial M} \, , \qquad p_a = P_a^0|_{\partial M} \, , \qquad \beta_a^{k} = P_a^{k}|_{\partial M} \, ,
$$
we can rewrite \eqref{Sepsilon} as
$$
S_\epsilon (\chi) = \int_{-\epsilon}^0 \d t \int_{\partial M} \mathrm{vol}_{\partial M}[ p_a \dot{\varphi}^a + \beta_a^k\partial_k\varphi^a -H(\varphi^a,p_a,\beta_a^k)] \, .
$$
Letting $\langle p, \dot{\varphi} \rangle = \int_{\partial M} p_a \dot{\varphi}^a \,\,  \mathrm{vol}_{\partial M}$ denote, as in \eqref{pairing_cotangent}, the natural pairing and, similarly, 
$$\langle \beta, \mathrm{d}_{\partial M}\varphi \rangle = \int_{\partial M} \beta_a^k \partial_k \varphi^a \,  \mathrm{vol}_{\partial M},
$$
we can define a density function $\mathcal{L}$ as,
\begin{equation}\label{densityL}
\mathcal{L}(\varphi,\dot{\varphi},p,\dot{p},\beta,\dot{\beta})=\langle p,\dot{\varphi}\rangle + \langle\beta, \mathrm{d}_{\partial M}\varphi \rangle - \int_{\partial M} H(\varphi^a,p_a,\beta_a^k) \, \mathrm{vol}_{\partial M} \, ,
\end{equation}
and then
$$
S_\epsilon (\chi) = \int_{-\epsilon}^0 \d t \, \, {\mathcal{L}}(\varphi,\dot{\varphi},p,\dot{p},\beta,\dot{\beta}) \, .
$$

Notice again that because of the existence of the collar $U_\epsilon$ near the boundary and the assumed form of  $\mathrm{vol}_ {U_\epsilon}$,  the elements in the bundle $i^*P(E)$ have the form $\rho_a^0 \d u^a \wedge \mathrm{vol}_{\partial M} + \rho_a^k \d u^a \wedge \d t \wedge i_{\partial /\partial x^k}\mathrm{vol}_{\partial M}$ and, as discussed in Lemma \ref{decomposition}, the bundle $i^*P(E)$ over $i^*E$ is isomorphic to the product $\bigwedge_1^m(i^*E) \times B$, where $B = \bigwedge_1^{m-1}(i^*E)$.   The space of double sections $(\varphi,p)$ of the bundle $\bigwedge_1^m(i^*E) \to i^*E \to \partial M$ correspond to the cotangent bundle $T^*\mathcal{F}_{\partial M}$ and the double sections $(\varphi, \beta)$ of the bundle $B \to i^*E \to \partial M$ correspond to a new space of fields at the boundary denoted by $\mathcal{B}$.

We will introduce now the total space of fields at the boundary $\mathcal{M}$ which is the space of double sections of the iterated bundle $i^*P(E) \to i^*E \to \partial M$. Following the previous remarks it is obvious that $\mathcal{M}$ has the form,
$$\mathcal{M}= \mathcal{T}^*\mathcal{F}_{\partial M} \times_{\mathcal{F}_{\partial M}} \mathcal{B} = \{(\varphi, p, \beta)\} \, .$$

Thus the density function $\mathcal{L}$, Eq. \eqref{densityL}, is defined on the tangent space $T\mathcal{M}$ to the total space of fields at the boundary and could be called accordingly the boundary Lagrangian of the theory.

Consider the action $ A = \int_{-\epsilon}^0 \mathcal{L} \,\, \d t$ defined on the space of curves $\sigma\colon ( -\epsilon, 0] \to \mathcal{M}$.
If we compute $\mathrm{d}A$ we obtain a bulk term, that is, an integral on $(-\epsilon, 0]$, and a term evaluted at $\partial [-\epsilon,0] = \{-\epsilon, 0\}$. Setting the bulk term equal to zero, we obtain the Euler-Lagrange equations of this system considered as a Lagrangian system on the space $\mathcal{M}$ with Lagrangian function $\mathcal{L}$, 
\begin{equation}\label{Euler-Lagrange Equations 1}
 \frac{\d}{\d t} \frac{{\delta}{\mathcal{L}}}{{\delta}{\dot{{\varphi}^a}}}= \frac{{\delta}{\mathcal{L}}}{{\delta}{{\varphi}^a}} \, ,
 \end{equation}
 which becomes,
 \begin{equation}\label{pidot}
\dot{p_a}= -{\partial}_k{\beta}_a^k - \frac{{\partial}H}{{\partial}{\varphi}^a} \, .
\end{equation}
Similarly, we get for the fields $p$ and $\beta$:
$$
\frac{\d}{\d t}\frac{{\delta}{\mathcal{L}}}{{\delta}{\dot{p}_a}}=\frac{{\delta}{\mathcal{L}}}{{\delta}{p_a}} \, , \quad
 \frac{\d}{\d t}\frac{{\delta}{\mathcal{L}}}{{\delta}{\dot{{\beta}}_a^k}}=\frac{{\delta}{\mathcal{L}}}{{\delta}{{\beta}_a^k}}
 $$
that become respectively,
\begin{equation}\label{phidot}
\dot{\varphi}_a = \frac{\partial H}{\partial p_a} \, ,
\end{equation}
and, the constraint equation:
\begin{equation}\label{dconstraint}
\d_{\partial M}{\varphi}-\frac{\partial H}{{\partial}{{\beta}_a^k}}=0 \, .
\end{equation}
Thus, Euler-Lagrange equations in a collar $U_\epsilon$ near the boundary, can be understood as a system of evolution equations on $T^*\mathcal{F}_{\partial M}$ depending on the variables $\beta_a^k$, together with a constraint condition on the extended space $\mathcal{M}$.  The analysis of these equations, Eqs. \eqref{pidot}, \eqref{phidot} and \eqref{dconstraint}, is best understood in a presymplectic framework.

\subsection{The presymplectic picture at the boundary and constraints analysis}\

We will introduce now a presymplectic framework on $\mathcal{M}$ that will be helpful in the study of Eqs.  \eqref{pidot}-\eqref{dconstraint}.

Let $\varrho :\mathcal{M} \longrightarrow \mathcal{T}^*\mathcal{F}_{\partial M}$ denote the canonical projection $\varrho(\varphi,p,\beta)=(\varphi,p)$.  (See Figure \ref{diagram}.)
Let $\Omega$ denote the pull-back of the canonical symplectic form ${\omega}_{\partial M}$ on $\mathcal{T}^*\mathcal{F}_{\partial M}$ to $\mathcal{M}$, i.e., let $\Omega=\varrho^*\omega_{\partial M}$.   Note that the form $\Omega$ is closed but degenerate, that is, it defines a presymplectic structure on $\mathcal{M}$. An easy computation shows that the characteristic distribution $\mathcal{K}$ of $\Omega$, is given by
$$
\mathcal{K} = \ker\Omega= \mathrm{span} \left\{ \frac{\delta}{{\delta}{\beta}_a^k} \right\} \, .
$$

Let us consider the function defined on $\mathcal{M}$, 
$$
\mathcal{H}(\varphi,p,\beta)= -\langle \beta, d_{\partial M}\varphi \rangle + \int_{\partial M} H({\varphi}^a, p_a, {\beta}_a^k)\, \mathrm{vol}_{\partial M} \, .
$$

\begin{figure}[ht]
\centering
\includegraphics[width=10cm]{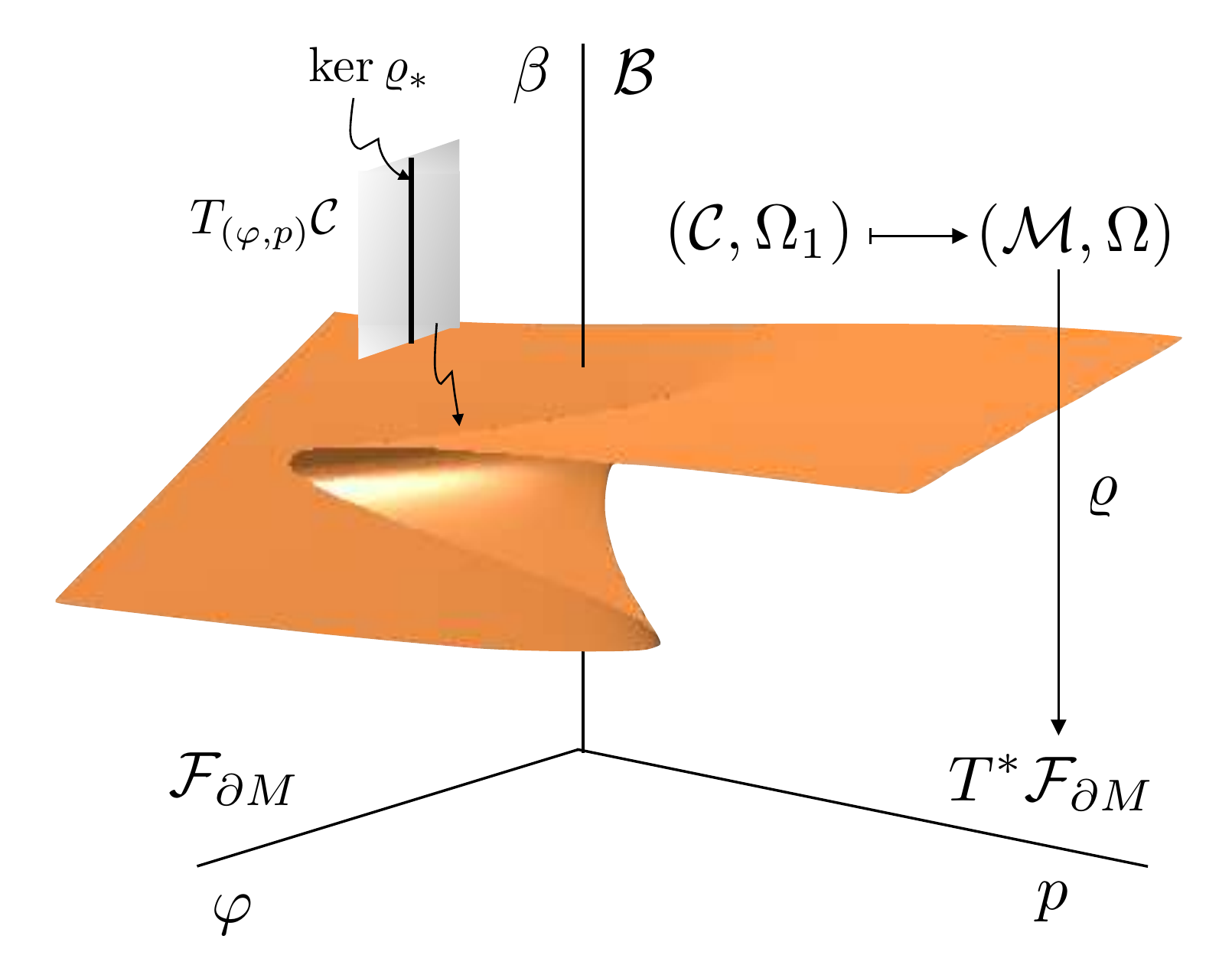}
\caption{The space of fields at the boundary $\mathcal{M}$ and its relevant structures.}\label{diagram}
\end{figure}

We will refer to $\mathcal{H}$ as the boundary Hamiltonian of the theory.
 Thus $\mathcal{L}$ can be rewritten as
$$\mathcal{L}(\varphi, \dot{\varphi},p,\dot{p},\beta, \dot{\beta})=\langle p,\dot{\varphi} \rangle - \mathcal{H}(\varphi,p,\beta)$$ and 
\begin{equation}\label{Lagrangian with Boundary Hamiltonian}
 S_{\epsilon}(\varphi, p, \beta) =\int_{-\epsilon}^0 [\langle p,\dot{\varphi} \rangle - \mathcal{H}(\varphi,p,\beta)] \d t \, ,
 \end{equation}
and therefore the Euler-Lagrange equations \eqref{phidot}, \eqref{pidot} and \eqref{dconstraint} can be written as
\begin{equation}\label{Hamilton's Evolution Equations}
 \dot{\varphi}^a = \frac{\delta \mathcal{H}}{\delta p_a}\, ,\qquad 
\dot{p}_a = -\frac{\delta \mathcal{H}}{\delta{\varphi}^a}  \, ,
\end{equation}
and
\begin{equation}\label{Hamilton's Constraint Equation}
 0= \frac{\delta\mathcal{H}}{\delta{\beta}_a^k} \, .
\end{equation}

Now it is easy to prove the following:

\begin{theorem}\label{presymplectic_equation}
The solutions to the equations of motion defined by the Lagrangian $\mathcal{L}$ over a collar $U_\epsilon$ at the boundary, $\epsilon$ small enough, are in one-to-one correspondence with the integral curves of the presymplectic system $(\mathcal{M},\Omega,\mathcal{H})$, i.e., with the integral curves of the vector field $\Gamma$ on $\mathcal{M}$ satisfying 
\begin{equation}\label{presymplectic_equation1}
i_\Gamma \Omega = \d\mathcal{H} \, .
\end{equation}
\end{theorem}

\begin{proof}  Let $\Gamma = A^a\frac{\delta}{\delta{\varphi}^a} + B^a\frac{\delta}{{\delta}p^a} + C^a\frac{\delta}{{\delta}{\beta}_a^k}$ be a vector field on $\mathcal{M}$ (notice that we are using an extension of the functional derivative notation introduced in Section \ref{sec:cotangent_boundary} on the space of fields $\mathcal{M}$).  Then because $\Omega = \delta{\varphi}^a \wedge \delta p_a$, we get from $i_{\Gamma}\Omega= \d\mathcal{H}$ that,
$$
A^a = \frac{{\delta}{\mathcal{H}}}{{\delta} p_a},\qquad
B^a = -\frac{{\delta}{\mathcal{H}}}{{\delta}{\varphi}^a} \, ,\qquad 0 = \frac{{\delta}{\mathcal{H}}}{{\delta}{\beta}_a^k} \, .
$$

Thus, $\Gamma$ satisfies Eq. \eqref{presymplectic_equation1} iff 
$$
\dot{\varphi}^a =\frac{{\delta}{\mathcal{H}}}{{\delta}p_a}, \qquad
\dot{p}_a = -\frac{{\delta}{\mathcal{H}}}{{\delta}{\varphi}^a} \, , \quad 
\mathrm{and} \quad 
0= \frac{{\delta}{\mathcal{H}}}{{\delta}{\beta}_a^k} \, .
$$
\end{proof}

Let us denote by $\mathcal{C}$ the submanifold of the space of fields $\mathcal{M} =T^*\mathcal{F}_{\partial M} \times \mathcal{B}$ defined by eq. $(3.9)$.  It is clear that the
restriction of the solutions of the Euler-Lagrange equations on $M$ to the boundary ${\partial}M$, are contained in $\mathcal{C}$; i.e., $ \Pi (\mathcal{EL})\subset\mathcal{C}.$

Given initial data $\varphi, p$ and fixing $\beta$, existence and uniqueness theorems for initial value problems when applied to the initial value problem above, would show the existence of solutions for small intervals of time, i.e., in a collar near the boundary.   

However, the constraint condition given by eq. \eqref{Hamilton's Constraint Equation}, satisfied automatically by critical points of $S_\epsilon$ on $U_\epsilon$, must be satisfied along the integral curves of the system, that is, for all $t$ in the neighborhood $U_\epsilon$ of $\partial M$.   This implies that consistency conditions on the evolution  must be imposed.  Such consistency conditions are just that the constraint condition eq. \eqref{Hamilton's Constraint Equation}, is preserved under the evolution defined by eqs. \eqref{Hamilton's Evolution Equations}. This is the typical situation that we will find in the analysis of dynamical problems with constraints and that we are going to summarily analyze in what follows.


\subsubsection{The Presymplectic Constraints Algorithm (PCA)}

Let $i$ denote the canonical immersion $\mathcal{C}=\{(\varphi,p,\beta)| \frac{\delta {\mathcal{H}}}{\delta{\beta}}=0\}\to \mathcal{M}$ and consider the pull-back of $\Omega$
to $\mathcal{C}$, i.e., $\Omega_1 = i^*\Omega$. Clearly then, $\ker \Omega_1 = \ker \varrho_* \cap T\mathcal{C}$.
But $\mathcal{C}$ is defined as the zeros of the function $\delta \mathcal{H}/\delta \beta$. Therefore if $\delta^2 \mathcal{H}/\delta^2\beta$ is nondegenerate (notice that the operator $\delta^2 \mathcal{H}/\delta
{\beta}_a^i{\delta}{\beta}_b^j$ becomes the matrix $\partial^2 H /\partial \beta_a^i \partial\beta_b^j$), by an appropriate extension of the Implicit Function Theorem, we could solve $\beta$ as a function of $\varphi$ and $p$.  In such case, locally, $\mathcal{C}$ would be the graph of a function $F\colon T^*\mathcal{F}_{\partial M} \to \mathcal{B}$, say $\beta = F(\varphi, p)$. This is precisely the situation we will see in the simple example of scalar fields in the next section.
Collecting the above yields:

\begin{proposition}  The submanifold $(\mathcal{C}, \Omega_1)$ of $(\mathcal{M},\Omega,\mathcal{H})$ is symplectic iff $H$ is regular, i.e., 
$\partial^2 H /\partial \beta_a^i \partial\beta_b^j$ is non-degenerate.  In such case the projection $\varrho$ restricted to $\mathcal{C}$, which we denote by $\varrho_C$, is a local symplectic diffeomorphism and therefore $\varrho_C^*\omega_{\partial M} = \Omega_1$.
\end{proposition}

When the situation is not as described above,
and $\beta$ is not a function of $\varphi$ and $p$, then $(\mathcal{C},\Omega_1)$ is indeed a presymplectic submanifold of $\mathcal{M}$ and $i_{\Gamma}{\Omega}=d\mathcal{H}$ will not hold necessarily at every point in $\mathcal{C}$.  In this case 
 we would apply Gotay's Presymplectic Constraints Algorithm [Go78], to obtain  the maximal submanifold of $\mathcal{C}$ for which 
 $i_{\Gamma}{\Omega}=d\mathcal{H}$ is consistent and that can be summarized as follows.


Consider a presymplectic system $(\mathcal{M}, \Omega, \mathcal{H})$ where $\mathcal{M}= T^*\mathcal{F}_{\partial M}\times\mathcal{B}$ and, $\Omega$ and $\mathcal{H}$ are as defined above. Let $\mathcal{M}_0 = \mathcal{M}$, $\Omega_0 = \Omega$, $\mathcal{K}_0 = \ker \Omega_0$, and $\mathcal{H}_0 = \mathcal{H}$.  We define the primary constraint submanifold $\mathcal{M}_1$ as the submanifold defined by the consistency condition for the equation $i_\Gamma \Omega_0 = \d\mathcal{H}_0$, i.e., 
$$
\mathcal{M}_1  = \{ \chi \in \mathcal{M}_0 \mid \langle Z_0(\chi) , \d\mathcal{H}_0(\chi) \rangle = 0, \, \,  \forall Z_0 \in \mathcal{K}_0 \} \, .
$$
 Thus $\mathcal{M}_1= \mathcal{C}$. Denote by $i_1 \colon \mathcal{M}_1 \to \mathcal{M}_0$ the canonical immersion. Let $\Omega_1 = i_1^*\Omega_0$, $\mathcal{K}_1 = \ker \Omega_1$, and $\mathcal{H}_1 = i_1^*\mathcal{H}_0$.   We now define recursively the $(k+1)$-th constraint submanifold as the consistency condition for the equation $i_\Gamma \Omega_k = \d\mathcal{H}_k$, that is,
$$
\mathcal{M}_{k+1}  = \{ \chi \in \mathcal{M}_k \mid \langle Z_k(\chi) , \d\mathcal{H}_k(\chi) \rangle = 0, \, \,  \forall Z_k \in \mathcal{K}_k \} \, \qquad k \geq 1 \, ,
$$
and $i_{k+1}\colon \mathcal{M}_{k+1} \to \mathcal{M}_k$ is the canonical embbeding (assuming that $\mathcal{M}_{l+1}$ is a regular submanifold of $\mathcal{M}_k$), and  $\Omega_{k+1} = i_{k+1}^*\Omega_k$, $\mathcal{K}_{k+1} = \ker \Omega_{k+1}$ and $\mathcal{H}_{k+1} = i_{k+1}^*\mathcal{H}_k$.

The algorithm stabilizes if there is an integer $r> 0$ such that 
$\mathcal{M}_{r} = \mathcal{M}_{r+1}$.
We refer to this $\mathcal{M}_r$ as the final constraints submanifold and we denote it by $\mathcal{M}_\infty$. Letting $i_\infty\colon \mathcal{M}_\infty \to \mathcal{M}_0$ denote the canonical immersion, we define,
$$
\Omega_\infty = i_\infty^*\Omega_0, \qquad \mathcal{K}_\infty = \ker \Omega_\infty\, , \qquad \mathcal{H}_\infty = i_\infty^*\mathcal{H}_0 \, .
$$
Notice that the presymplectic system $(\mathcal{M}_\infty, \Omega_\infty,  \mathcal{H}_\infty )$ is always consistent, that is, the dynamical equations defined by $i_\Gamma \Omega_\infty = d\mathcal{H}_\infty$ will always have solutions on $\mathcal{M}_\infty$.  The solutions will not be unique if $\mathcal{K}_\infty \neq 0$, hence the integrable distribution $\mathcal{K}_\infty$ will be called the ``gauge'' distribution of the system, and its sections (that will necessarily close a Lie algebra), the ``gauge'' algebra of the system.\\
    
In the particular theories considered in this work we found that $\mathcal{M}_{\infty}=\mathcal{M}_1=\mathcal{C}$ and we do not needed to go beyond the first step of the algorithm to obtain the final constraints submanifold.\\ 

The quotient space $\mathcal{R} = \mathcal{M}_\infty / \mathcal{K}_\infty$, provided it is a smooth manifold,  inherits a canonical symplectic structure $\omega_\infty$ such that $\pi_\infty^*\omega_\infty = \Omega_\infty$, where $\pi_\infty \colon \mathcal{M}_\infty \to \mathcal{R}$ is the canonical projection. We will refer to it as the reduced phase space of the theory.   Notice that the Hamiltonian $\mathcal{H}_\infty$ also passes to the quotient and we will denote its projection by $h_\infty$ i.e., $\pi_\infty^* h_\infty = \mathcal{H}_\infty$.  

Thus the Hamiltonian system $(\mathcal{R}, \omega_\infty, h_\infty)$ will provide the canonical picture of the theory at the boundary and its quantization will describe the states and dynamics of the theory with respect to observers sitting at the boundary $\partial M$.   

 Of course all the previous constructions depend on the boundary $\partial M$ of the manifold $M$.  For instance, if we assume that $M$ is a globally hyperbolic space-time of the form $M \cong [t_0,t_1] \times \Sigma$, then $\partial M = \{t_0\}\times \Sigma \cup \{t_1\}\times \Sigma$.  But if we use a different Cauchy surface $\Sigma'$, the boundary of our space-time will vary and we will get a new reduced phase space $(\mathcal{R}', \omega', h')$ for the theory.  However in this case it is easy to show that there is a canonical symplectic diffeomorphism $S\colon \mathcal{R} \to \mathcal{R}'$ such that $h = S^*h'$. (Recall that in such case there will exist a canonical diffeomorphism $\Sigma \to \Sigma'$ that will eventually induce the map $S$ above.)

Recall that $\Pi (\mathcal{EL}) \subset \mathcal{C}$. We easily show then that after the reduction to $\mathcal{R}$, the reduced submanifold of boundary values of Euler-Lagrange solutions of the theory, $\widetilde{\Pi}(\mathcal{E}\mathcal{L})$, 
is an isotropic submanifold, now of the reduced phase space.   


\begin{theorem}\label{reduction_theorem}
The reduction $\widetilde{\Pi}(\mathcal{EL})$  of the submanifold of Euler-Lagrange fields of the theory is an isotropic submanifold of the reduced phase space $\mathcal{R}$ of the theory.
\end{theorem}

\begin{proof}      
It is clear that $\Pi(\mathcal{EL}) \subset \Pi(\mathcal{EL}_\epsilon) \subset \mathcal{M}_{\infty}$ where $\mathcal{EL}_\epsilon = \mathcal{EL}_{U_\epsilon}$ are 
the critical points of the action $S_\epsilon$, i.e., solutions of the Euler-Lagrange equations of the theory on $U_\epsilon$.

 The reduction $\widetilde{\Pi}(\mathcal{EL}) = \Pi(\mathcal{EL})/ (\mathcal{K}_\infty\cap T\, \Pi (\mathcal{EL}))$ of the isotropic submanifold  $ \Pi(\mathcal{EL})$ to the reduced phase space $\mathcal{R} = \mathcal{M}_\infty / \mathcal{K}_\infty$ is isotropic because $\pi_\infty^*\omega_\infty = \Omega_\infty$, hence $\pi_\infty^* (\omega_\infty\mid_{\widetilde{\Pi}(\mathcal{EL})}) = (\pi_\infty^* \omega_\infty)\mid_{\Pi(\mathcal{EL})} = \varrho^*\mathrm{d} \alpha \mid_{\Pi(\mathcal{EL})} = 0$. 
\end{proof}


\subsection{Reduction at the boundary and gauge symmetries}\label{reduction_gauge}

If our theory $(P(E), \Theta_H)$ has $\mathcal{G}$ as a covariant symmetry group, then because of Noether's theorem, Thm. \ref{Noether}, and Eq. \eqref{boundary_conservation}, we have that $ J_\xi[\chi]$, with $\Pi (\chi) = (\varphi, p)$ a closed $(m-1)$-form. Hence $\int_{\partial M} i^*J_\xi[\chi] = 0$, and so 
$$
\langle \mathcal{J}(\varphi, p), \xi \rangle = \int_{\partial M}  i^*J_\xi[\chi] = 0 \, .
$$
Then $\mathcal{J}(\Pi (\chi )) = 0$, and therefore,
$$
\Pi (\mathcal{EL}) \subset \mathcal{J}^{-1} (\mathbf{0}) \, .
$$ 

  There is a natural reduction of the theory at the boundary defined by the covariant symmetry $\mathcal{G}$ for the following reason:  Provided that the value $\mathbf{0}$ of the moment map $\mathcal{J}$ is weakly regular, the submanifold $\mathcal{J}^{-1}(\mathbf{0}) \subset T^*\mathcal{F}_{\partial M}$ is a coisotropic submanifold and the characteristic distribution $\ker i_0^*\omega_{\partial M}$ of the pull-back of the canonical symplectic form on $T^*\mathcal{F}_{\partial M}$ to it, is the distribution defined by the orbits of the group $\mathcal{G}_{\partial M}$.  From Prop. \ref{boundary_moment},  $\mathcal{J}$ is the moment map of the canonical lifting of the action of the group $\mathcal{G}_{\partial M}$ on $\mathcal{F}_{\partial M}$.
  
 From the above and by Thm \ref{presymplectic_equation}, Lemma $3.4$ follows easily. 

\begin{lemma}
$$
\mathcal{J}^{-1}(\mathbf{0})  \subset   \varrho (\mathcal{M}_\infty)\, .
$$
\end{lemma}

\begin{proof}   If $\mathcal{G}$ is a symmetry group of the Hamiltonian $H$ of the theory, then it is clear that $\mathcal{G}_{\partial M}$ is a symmetry group of the function $\mathcal{H}$, with the canonical action of $\mathcal{G}_{\partial M}$ on the total space of fields at the boundary $\mathcal{M}$.   

 Then if $\zeta \in \mathcal{M}_\infty$, there exists $\Gamma$ at $\zeta$ such that $i_\Gamma \Omega_\infty = \d\mathcal{H}_\infty$ and the integral curve $\gamma$ of $\Gamma$ passing through $\zeta$ lies in $\mathcal{M}_\infty$.  But $\varrho(\gamma) \subset \mathcal{J}^{-1}(\mathbf{0})$, because it is the projection of an integral curve of a solution of Euler-Lagrange equations in $U_\epsilon$. But because the Hamiltonian $\mathcal{H}$ is invariant, the trajectory must lie in a level surface of the moment map $\mathcal{J}$. Hence $\mathcal{J}^{-1}(\mathbf{0}) \subset \varrho (\mathcal{M}_\infty)$.
\end{proof}

 Because, $\mathcal{R} = \mathcal{M}_\infty /\mathcal{K}_\infty$ and $\ker \varrho_* \cap T\mathcal{M}_\infty \subset \mathcal{K}_\infty$, we get that $\mathcal{M}_\infty /\mathcal{K}_\infty \cong \varrho (\mathcal{M}_\infty) / \varrho_*(\mathcal{K}_\infty)$.   Now if we are in the situation where $\varrho(\mathcal{M}_\infty) = \mathcal{J}^{-1}(\mathbf{0})$, then $\mathcal{R} \cong \varrho (\mathcal{M}_\infty) / \varrho_*(\mathcal{K}_\infty) = \mathcal{J}^{-1}(\mathbf{0}) / \ker \omega_{\partial M}\mid_{\mathcal{J}^{-1}(\mathbf{0})}$.   Hence
because of the standard Marsden-Weinstein reduction theorem the reduced phase space of the theory is obtained simply as,
\begin{equation}\label{MWR}
\mathcal{R} \cong  \mathcal{J}^{-1}(\mathbf{0})/ \mathcal{G}_{\partial M} \, .
\end{equation}


\subsection{A simple example: the scalar field}\label{scalar}

We will consider the simple example of a real scalar field on a globally hyperbolic space-time $(M, \eta)$ of dimension $m = 1 +d$ with boundary $\partial M$ a Cauchy surface and hence $M \cong (-\infty, a] \times \partial M$.  The configuration fields of the system are sections of the (real) line bundle $\pi \colon E \to M$, where $\pi$ is projection onto the first factor.  Bundle coordinates will have the form $(x^\mu,u)$, $\mu=0,1,...,d$.

If the bundle $E \to M$ were trivial, $E \cong M \times \mathbb{R}$, the first jet bundle $J^1E$ would be the affine bundle $J^1E \cong T^*M \times \mathbb{R} \rightarrow E$ with bundle coordinates $(x^\mu,u; u_\mu)$, $\mu = 0,1,...,d$.  The covariant phase space $P(E)$, in such case, would be isomorphic to $TM\times \mathbb{R}$ with bundle coordinates  $(x^{\mu},u; \rho^\mu)$.

As explained in Section \ref{sec:multisymplectic}, by using the volume form $\mathrm{vol}_M = \sqrt{| \eta |} \, \d^mx$ defined by the metric $\eta$ (in arbitrary local coordinates $x^\mu$), elements in $P(E)$ can be identified with semi-basic $m$-forms on $E$, $w \in\bigwedge^m_1(E)$, $w = \rho^\mu \d u \wedge \mathrm{vol}_\mu^d + \rho_0 \mathrm{vol}_M$, $\mathrm{vol}_\mu^d = i_{\partial/\partial x^\mu} \mathrm{vol}_M$, after we mod out basic $m$-forms, $\rho_0 \mathrm{vol}_M$.

The space of fields in the bulk, $\mathcal{F}_{P(E)}= \{ \chi = ( \Phi , P)\}$ consists of double sections of the iterated bundle $P(E) \to E \to M$, $\Phi \colon M \to E$, $u = \Phi (x)$, and $P\colon E \to P(E)$, $\rho = P(u)$ that, in the instance of a trivial bundle $E$, can be described as maps $\Phi \colon M \to \mathbb{R}$, the configuration fields, and $(m-1)$-forms, $P = P^{\mu}(x)\mathrm{vol}^d_\mu$, the momenta fields.

The Hamiltonian $H$ of the theory determines a section of the projection $M(E) \to P(E)$ by fixing the variable $\rho_0$ above, i.e., $\rho_0 = -H(x^\mu,u, \rho^\mu)$. One standard choice for $H$ in such case is:
$$
H(x^\mu,u; \rho^\mu) = \frac{1}{2}\eta_{\mu\nu}\rho^\mu  \rho^\nu + V(u) \, ,
$$ 
with $V(u)$ a smooth function on $\mathbb{R}$.   The particular instance of $V(u) = m^2u^2$ gives us the Klein-Gordon system.  

The canonical $m$-form $\Theta$ in $\bigwedge^m_1(E)$ can be pulled back to $P(E)$ along $H$ and takes the form,
$$
\Theta_H = \rho^{\mu} du \wedge \mathrm{vol}^d_{\mu} - H (u) \mathrm{vol}_M \, . 
$$

With the above choice for $H$, the action functional of the theory becomes:
\begin{equation}
S(\Phi,P) =  \int_M \left[ P^\mu(x) \partial_\mu\Phi (x) - \frac{1}{2} \eta_{\mu \nu} P^\mu P^\nu - V(\Phi) \right] \sqrt{|\eta |}\, \d^mx \, .
\end{equation}

The space of boundary fields  $T^*\mathcal{F}_{\partial M} = \{ (\varphi , p) \}$ is given by $\varphi = \Phi\mid_{\partial M}$, $p = P^0\mid_{\partial M}$.
Computing the differential of the action we get,
\begin{eqnarray*}
\mathrm{d} S_{(\Phi, P)} ( \delta \Phi, \delta P ) &=& \int_M  [  \delta P^\mu ( \partial_\mu\Phi  - 
\eta_{\mu \nu} P^\nu) + \delta \Phi (-\frac{1}{\sqrt{|\eta|}} \partial_\mu (P^\mu\sqrt{|\eta |})  \\ &-& V'(\Phi))] \sqrt{|\eta |}\, \d^mx + \int_{\partial M}p \delta \varphi \, \mathrm{vol}_{\partial M} \, ,
\end{eqnarray*}
and the Euler-Lagrange equations of the theory are given by,
\begin{equation}\label{scalar_EL}
\frac{1}{\sqrt{|\eta |}}{\partial}_{\mu}(P^{\mu}\sqrt{|\eta |})+ V'(\Phi) =0 \, , \qquad {\partial}_\mu \Phi - \eta_{\mu\nu}P^\nu = 0 \, 
.
\end{equation}

From the second of the Euler equations we get, $P^\nu = \eta^{\mu\nu} \partial_\mu \Phi$, and
substituting into the first we get 
\begin{equation}\label{Laplace}
\frac{1}{\sqrt{|\eta |}} \partial_\mu(\sqrt{|\eta|} \eta^{\mu\nu} \partial_\nu\Phi)= -V'(\Phi) \, .
\end{equation}
The first term is the Laplace-Beltrami operator of the metric $\eta$, i.e., the d'Alembertian in the case of the Minkowski metric.\\
 
Note that had we instead chosen normal local coordinates on $M$, the volume element in such charts would take the form
 $$
\mathrm{vol}_M = \d x^0 \wedge \d x^1 \wedge \cdots \wedge \d x^d $$
  and then equations $(3.13)$ would just be Hamilton's equations:
 
$$ 
\partial_\mu P^\mu = -\frac{\partial H}{ \partial \Phi } \, \qquad \partial_\mu\Phi = \frac{\partial H}{\partial P^\mu} \, .
$$


\subsubsection*{The evolution picture near the boundary}\label{sec:scalar_boundary}
We consider a collar around the boundary $U_\epsilon = (- \epsilon,0] \times \partial M$ with coordinates $t=x^0$ and $x^i$, $i = 1, \ldots, d$.  We assume that $\eta=-\d t^2 + \eta_{0i}(x)\d t \otimes \d x^i + g_{ij}(x)\d x^i \otimes \d x^j$ and $g = g_{ij}(x) \d x^i\otimes \d x^j$ defines a Riemannian metric on $\partial M$.   Writing again the action functional S restricted to fields $\Phi , P$ defined on $U_{\epsilon}$, we have,
$$
S_\epsilon (\Phi , P) =  \int_{-\epsilon}^0 \d t \int_{\partial M} \mathrm{vol}_{\partial M} \sqrt{|\eta |} (P^0 \partial_0 \Phi + P^i \partial_i \Phi - \frac{1}{2} \eta_{\mu\nu} P^\mu P^\nu - V(\Phi)) \, .
$$
Consider the fields at the boundary $\varphi$ and $p$ defined before and
 $\beta^i = P^i\mid_{\partial M}$.
Also, let $\Delta = \sqrt{|\eta|}/\sqrt{|g|}$.  Then $\sqrt{|\eta |} d^mx = \Delta \, \d t \wedge \mathrm{vol}_{\partial M}$.

Therefore we can write,
\begin{eqnarray*}
S_{\epsilon}({\Phi},P) &=& \int_{-\epsilon}^0 \d t \int_{{\partial}M}\, \mathrm{vol}_{\partial M} \, \Delta\,  [ p \dot{\varphi} + \beta^i \partial_i \phi + \frac{1}{2} p^2 - \eta_{0i} p \beta^i - \frac{1}{2} g_{ij} \beta^i \beta^j - V(\phi)] 
\\ &=& \int_{-\epsilon}^0 \d t\, \, \left[ \langle p, \dot{\varphi} \rangle - \mathcal{H}(\varphi,p,\beta )  \right]\, 
\end{eqnarray*}
where
\begin{equation}\label{Deltaproduct}
\langle p, \dot{\varphi} \rangle = \int_{\partial M} \,  p(x) \dot{\varphi}(x)  \Delta\, \mathrm{vol}_{\partial M} \, ,
\end{equation}
denotes the scalar product on functions on $\partial M$ defined by the volume $\Delta\, \mathrm{vol}_{\partial M}$,
and $\mathcal{H} \colon \mathcal{M} \to \mathbb{R}$ denotes the Hamiltonian function induced from the Hamiltonian $H$ of the theory,
$$
\mathcal{H} (\varphi,p,\beta) =  -\langle \beta, d_{\partial M} \varphi \rangle - \frac{1}{2} \langle p, p \rangle + \langle p, \tilde{\beta} \rangle + \frac{1}{2} \langle \beta,\beta \rangle + \int_{\partial M} V(\varphi) \Delta\, \mathrm{vol}_{\partial M} \, ,
$$
with $\tilde{\beta} = \eta(d/\d t,\beta) = \eta_{i0}\beta^i$, $\langle p, p \rangle$ and $\langle p, \tilde{\beta} \rangle$ defined as in eq. \eqref{Deltaproduct}.  The product $\langle \beta,\beta \rangle$ denotes the scalar product of vector fields defined by the metric $g$, i.e.,
$$
\langle \beta,\beta \rangle = \int_{\partial M} g_{ij} \beta^i(x) \beta^j(x) \Delta\, \mathrm{vol}_{\partial M} \, ,
$$
and $\langle \beta, d_{\partial M} \varphi \rangle$ is the natural pairing between vector fields and 1-forms on $\partial M$, that is
$$
\langle \beta,d_{\partial M} \rangle = \int_{\partial M} \beta^i(x) \partial_i \varphi(x) \Delta\, \mathrm{vol}_{\partial M} \, .
$$

As in Section \ref{sec:dynamical_eqs} we denote the space of all fields at the boundary, 
the dynamical fields $\varphi$, $p$ and the fields $\beta^i$, as 
$\mathcal{M} = T^*\mathcal{F}_{\partial M} \times \mathcal{B} = \{ (\varphi, p; \beta ) \}$ and
Hamilton's equations for $\mathcal{H}$ are given by,
$$
\dot{\varphi} = \frac{\delta \mathcal{H}}{\delta p} = -p + \tilde{\beta}   \, , \qquad
\dot{p} = - \frac{\delta \mathcal{H}}{\delta \varphi} = -V'(\varphi)-\mathrm{div\,}\beta  \, , 
$$
together with the constraint equation obtained from the variation of $S_\epsilon$ with respect to $\beta$,
$$
0 = \frac{{\delta}\mathcal{H}}{{\delta}{\beta}^i} = -\partial_i \phi + \eta_{0i}p + g_{ij} \beta^j \, .
$$
Thus we get, 
\begin{eqnarray*} 
\dot{p} &=& - \mathrm{div\,}\beta - V'(\varphi) \\ 
\dot{\varphi} &=& - p + \tilde{\beta} 
\end{eqnarray*}
and the constraints equations,
\begin{equation}\label{scalar_constraint}
- \d_{\partial M} \varphi + p^\flat + \beta^\flat = 0 \, ,
\end{equation}
where $\beta^\flat = g(\beta, \cdot)$ is the 1-form associated to the vector $\beta$ by the metric $g$, and $p^\flat$ is the 1-form associated to the vector $p\, \partial / \partial t$.
   
  Let $\mathcal{C}= \{ (\varphi,p,\beta)\in\mathcal{M} \mid \delta \mathcal{H}/ \delta \beta = 0 \}$, the submanifold of $\mathcal{M}$ defined by the constraints \eqref{scalar_constraint}, and let $\varrho \colon \mathcal{M} \to T^*\mathcal{F}_{\partial M}$ denote the canonical projection.  We can solve for $\beta^i$ as a function of $\varphi$ and $p$ in the constraint equation \eqref{scalar_constraint}, obtaining
 $\beta^j = g^{ij}(\partial_i \varphi - g_{0i}p)$ or, more intrisically,
 $$
 \beta = \d_{\partial M}\varphi^\sharp - p \frac{\partial}{\partial t} \, ,
 $$
 where $d_{\partial M}\varphi^\sharp$ is the vector field associated to the 1-form $d_{\partial M}\varphi$ by means of the metric $g$.
Thus the restriction of $\varrho$ to $\mathcal{C}$ is a diffeomorphism onto $T^*\mathcal{F}_{\partial M}$. If we denote by $\Omega$ the pull-back $\varrho^*\omega_{\partial M}$ to $\mathcal{M}$ of the canonical symplectic form on $T^*\mathcal{F}_{\partial M}$ and by $\Omega_{\mathcal{C}}$ its restriction to the submanifold $\mathcal{C}$,
the restriction of the canonical projection $\varrho \colon \mathcal{M} \to T^*\mathcal{F}_{\partial M}$ to $\mathcal{C}$ provides a symplectic diffeomorphism $(\mathcal{C}, \Omega_\mathcal{C}) \cong (T^*\mathcal{F}_{\partial M}, \omega_{\partial M})$.  

Moreover, the projection $\Pi (\mathcal{EL})$ of the space of solutions to the Euler-Lagrange equations \eqref{scalar_EL}  to the boundary, defines, wherever it is a smooth submanifold, an isotropic submanifold of $T^*\mathcal{F}_{{\partial}M}$, as shown in Thm. \ref{Lagrangian_sub}.    
 $\Pi (\mathcal{EL})$ is not necessarily a Lagrangian submanifold because in general the Dirichlet problem defined by boundary conditions $(\varphi, p)$ for Eq. \eqref{Laplace} doesn't have a solution.    The situation is different in the Euclidean case, i.e., if $(M, \eta)$ is a Riemannian manifold, the Laplace-Beltrame operator would be elliptic and the Dirichlet problem would always have a unique solution.  In such case the space $\Pi (\mathcal{EL})$ would certainly be a Lagrangian submanifold of $T^*\mathcal{F}_{\partial M}$.

\subsection{Another example: The Poisson $\sigma$-model}

We will illustrate the previous ideas as they apply to the case of the Poisson $\sigma$-model.
We note that the Poisson $\sigma$-model  (P$\sigma$M for short) was analyzed in depth by A. Cattaneo \emph{et al}  \cite{Ca00} and provides a quantum field theory interpretation of Konsevitch's quantization of Poisson structures. We will just concentrate on its first order covariant Hamiltonian formalism along the lines described earlier in this paper.   

We will consider a Riemann surface $\Sigma$ with smooth boundary $\partial \Sigma \neq \emptyset$.  We may assume that $\Sigma$ also carries a Lorentzian metric. This will not play a significant role in the discussion and we can stick to a Euclidean picture by selecting a Riemannian metric on $\Sigma$.  Local coordinates on $\Sigma$ will be denoted as always by $x^\mu$, $\mu = 0,1$.  

Let $(P, \Lambda)$ be a Poisson manifold with local coordinates $u^a$, $a = 1, \ldots, r$.   The Poisson tensor $\Lambda$ will be expressed in local coordinates as
$$   
\Lambda = \Lambda^{ab}(y) \frac{\partial}{\partial u^a} \wedge \frac{\partial}{\partial u^b} \, , 
$$
and it defines a Poisson bracket on functions $f,g$ on $P$,
$$
\{ f, g \} = \Lambda (\d f, \d g) \, .
$$
The bundle $E$ of the theory, will be the trivial bundle $E = \Sigma \times P$ with projection $\pi$, the canonical projection on the first factor.  The first jet bundle $J^1E$ is the affine bundle over $E$ modeled on $VE \otimes T^*\Sigma$, however in this case, because of the triviality of $E$, we have that $VE \cong TP$ and the affine bundle is trivial.  Now the dual bundle $ P(E)$ will be naturally identified with the vector bundle over $E$ modeled on $T^*P \otimes T\Sigma$, that is, its sections will be vector fields on $\Sigma$ with values on 1-forms on $P$. However as shown in the general case, we may use a volume form $\mathrm{vol}_\Sigma$ on $\Sigma$ (for instance that provided by a Riemannian metric) to identify elements on $P(E)$ with 1-semibasic forms on $E$, i.e.
$$
P = P_a^\mu \,  \d u^a \wedge i_{\partial / \partial x^\mu} \mathrm{vol}_\Sigma \, ,
$$
and the corresponding double sections $\chi = (\Phi, P)$ of $P(E) \to E \to \Sigma$, with 1-forms $\eta$ on $\Sigma$ with values on 1-forms on $P$ along the map $\Phi \colon \Sigma \to P$, that is,
$$
P \colon T\Sigma \to T^*P \, , \qquad \tau_P \circ P = \Phi \, .
$$
The covariant Hamiltonian of the theory will be given by,
$$
H ( x, u ; P) = \frac{1}{2} \Lambda^{ab}(u) (P_a^\mu, P_b^\nu) \epsilon_{\mu\nu} \, 
$$
with $\mathrm{vol}_{\Sigma} = \epsilon_{\mu\nu} \d x^\mu\wedge \d x^\nu$.
The action of the theory is thus
\begin{equation}\label{action_PsM}
S_P (\chi ) = \int_\Sigma \chi^*\Theta_H = \int_{\Sigma} \left[ P_a^\mu \partial_\mu \Phi^a - H \right] \mathrm{vol}_{\Sigma}.
\end{equation}
Notice that $P_a = P_a^\mu \d x_\mu$ and that $\d x_\mu = i_{\partial / \partial x^\mu} \mathrm{vol}_\Sigma$ is a 1-form on $\Sigma$.  $H \mathrm{vol}_{\Sigma}$ can be expressed as
$$
H ( x, u ; P) \mathrm{vol}_{\Sigma}= \frac{1}{2} \Lambda^{ab}(u) (P_a \wedge P_b)
$$
and the first term in the action becomes simply $P_a\wedge d\Phi^a$. Thus the action of the theory is simply given as
$$
S_P (\Phi, P) = \int_{\Sigma} P_a (x) \d\Phi^a (x) - \frac{1}{2} \Lambda^{ab}(\Phi (x)) (P_a(x) \wedge P_b(x)) \, ,
$$
or more succinctly, 
$$
S_P (\Phi, P) = \int_{\Sigma} \langle P \wedge \d\Phi \rangle - \frac{1}{2} (\Lambda \circ \Phi) (P \wedge P ) \, ,
$$
where $\langle \cdot, \cdot \rangle$ now denotes the natural pairing between $T^*P$ and $TP$.

To get the evolution picture of the theory near the boundary, we choose a collar $U_\epsilon \cong (-\epsilon, 0 ]\times \partial \Sigma$ around the boundary $\partial \Sigma$ and we expand the action $S_P$ of the theory, eq. \eqref{action_PsM} restricted to fields defined on $U_\epsilon$.    We obtain,
$$
S_{P, U_\epsilon} = \int_{-\epsilon} \d t \int_{\partial \Sigma} \d u \left[ p_a \dot{\varphi}^a + \beta_a \acute{\varphi}^a - \Lambda^{ab} p_a \beta_b \right] \, ,
$$
where the boundary fields $p_a$ and $\beta_a$ are defined as before,
$$
p_a = P_a^0 \mid_{\partial \Sigma} \, , \qquad \beta_a = P_a^1\mid_{\partial \Sigma} \, .
$$
The volume form and the coordinate $u$ along the boundary $\partial \Sigma$ have been chosen so that $\mathrm{vol_\Sigma} = \d t \wedge \d u$, and $\acute{\varphi}^a$ denotes $\partial \varphi^a /\partial u$.

As before, the cotangent bundle of boundary fields is $T^*\mathcal{F}_{\partial \Sigma}$ with the canonical form $\alpha = p_a \delta \varphi^a$.  In order to analyze the consistency of the Hamiltonian theory at the boundary, we introduce the extended phase space $\mathcal{M} = T^*\mathcal{F}_{\partial \Sigma} \times \mathcal{B}$, with its presymplectic structure $\Omega = \delta \varphi^a \wedge \delta p_a$ and the boundary Hamiltonian
$$
\mathcal{H} (\varphi, p, \beta) = - \beta_a \acute{\varphi}^a + \Lambda^{ab}(\varphi)p_a\beta_b \, .
$$
Solving for the Euler-Lagrange equations we obtain two evolution equations,

$$
\dot{\varphi}^aÊ= \frac{\delta \mathcal{H}}{\delta p_a} = \Lambda^{ab} \beta_b \, , \qquad \dot{p_a} = -\frac{\delta \mathcal{H}}{\delta \varphi^a} = - \acute{\beta}_a - \frac{\partial \Lambda^{bc}}{\partial \xi^a} p_b\beta_c \, , 
$$
and one constraint equation equation, 
\begin{equation}\label{constraint_PsM}
0 = \frac{\delta \mathcal{H}}{\delta \beta_a} = - \acute{\varphi}^a - \Lambda^{ab}(\varphi)p_b \, .
\end{equation}
Thus the first constraints submanifold $\mathcal{M}_1$ will be defined by eq. \eqref{constraint_PsM}.  Notice the constraint defining $\mathcal{M}_1$ does not depend on the fields $\beta^a$, thus $\mathcal{M}_1$ is a cylinder along the projection $\varrho$ over its projection $W = \varrho (\mathcal{M}_1) \subset T^* \mathcal{F}_{\partial M}$.

Notice that $\Omega = \varrho^*\omega_{\partial M}$ is such that $\ker \Omega = \mathcal{K} = \{ \delta/ \delta \beta^a \} $.  Thus, $\mathcal{K} \subset \ker \Omega_1$, where $\Omega_1$ is the restriction of $\Omega$ to $\mathcal{M}_1$.  It is easy to check that $\ker \Omega_1 = \mathcal{K} \oplus \ker \Omega_{\mathcal{C}}$, where $\Omega_{\mathcal{C}}$ is the pull-back of $\omega_{\partial M}$ to $\mathcal{C}$.

The submanifold $ W \subset T^* \mathcal{F}_{\partial M}$ is defined by the constraint 
$$
\Psi^a(\varphi, p) = - \acute{\varphi}^a - \Lambda^{ab}(\varphi ) p_b \, ,
$$
whose Hamiltonian vector field $X_a$, i.e. $X_a$ such that
$$
i_{X_a} \omega_{\partial M} = \d\Psi^a \, ,
$$
is given by
$$
X_a (\varphi, p) = \Lambda^{ab}(\varphi) \frac{\delta}{\delta \varphi^b} - \left( \partial_u \delta_a^c- p_b \frac{\partial \Lambda^{ab}}{\partial \varphi^c}  \right) \frac{\delta}{\delta p_c}  \, .
$$
 A simple computation shows that
$$
X_a (\Psi^b)\mid_{\mathcal{C}} = 0 \, .
$$
Hence $TW^\perp \subset TW$ and consequently, not only $W$, but also $\mathcal{M}_1$ are coisotropic submanifolds.( In describing $\mathcal{M}_1$ as a coisotropic submanifold of the presymplectic manifold $\mathcal{M}$ we mean simply that $T\mathcal{M}_1^\perp \subset T\mathcal{M}_1$.)

The stability of the constraints shows that the PCA algorithm stops at $\mathcal{M}_1$.   Then the reduced (or physical) phase space of the theory is
$$
\mathcal{R} = \mathcal{M}_1/\ker \Omega_1 \cong \mathcal{C}/ \mathrm{span}\{X_a \} \, .
$$
The reduced phase space is a symplectic manifold, that in this case happens to be finite-dimensional.

In some particular cases it can be computed explicitly (for instance $\Sigma = [0,1]\times [0,1]$ with appropriate boundary conditions).   In some instances it happens to inherit a groupoid structure that becomes the symplectic groupoid integrating the Poisson manifold $P$ \cite{Ca01}.


\section{Yang-Mills theories on manifolds with boundary as a covariant Hamiltonian field theory}\label{sec:Yang-Mills}


\subsection{The multisymplectic setting for Yang-Mills theories}\

Recall from the introduction, $(M, \eta)$ is an oriented smooth manifold of dimension $m = 1 + d$ with boundary $\partial M \neq \emptyset$. It carries either a Riemannian or a Lorentzian  metric $\eta$, in the later case of signature $(-+\cdots +)$ and such that the connected components of $\partial M$ are space-like submanifolds, that is, the restriction $\eta_{\partial M}$ of the Lorentzian metric to them is a Riemannian metric.   

Yang-Mills fields are principal connections $A$ on some principal fiber bundle $\rho \colon P \to M$ with structural group $G$.   For clarity in the exposition we are going to make the assumption that $P$ is trivial (which is always true locally), i.e., $P \cong M \times G  \to M$ where (again, for simplicity) $G$ is a compact semi-simple Lie group with Lie algebra $\mathfrak{g}$.

Under these assumptions, principal connections on $P$ can be identified with $\mathfrak{g}$-valued 1-forms on $M$, i.e., with sections of the bundle $E=T^*M \otimes \mathfrak{g} \longrightarrow M$.  Local bundle coordinates in the bundle $E \to M$ will be written as $(x^\mu, A_\mu^a)$, $\mu = 1, \ldots, m$, $a= 1, \ldots, \dim\mathfrak{g}$, where 
$A = A_\mu^a \xi_a \in \mathfrak{g}$ with ${\xi}_a$ a basis of the Lie algebra $\mathfrak{g}$.  Thus, a section of the bundle can be written as 
\begin{equation}\label{connectionA}
A(x) = A^a_{\mu}(x)\, \d x^{\mu}{\otimes}{\xi}_a \, .
\end{equation}
 
The covariant Hamiltonian formalism will be formulated by considering the bundle $P(E)$, the affine dual of the first jet bundle $J^1E$.   Let us recall from the general discussion on Sect. \ref{sec:general}, that $J^1E$ is an affine bundle modeled on the vector bundle $T^*M\otimes VE \cong T^*M\otimes T^*M \otimes \mathfrak{g}$.   The affine dual of $J^1E$ can thus be modeled on the vector bundle $TM\otimes TM \otimes \mathfrak{g}^*$.   

The multisymplectic formalism is described in the manifold $P(E)$ whose elements can be identified with 1-semibasic $m$-forms
$$
P =  P^{\mu \nu}_a{\d A}^a_{\mu} \wedge \d^{m-1}x_{\nu} \, ,
$$
where $\\d^{m-1}x_\nu = i_{\partial/\partial x^\nu} \mathrm{vol}_\eta$ and $\mathrm{vol}_M$ is
the canonical volume form on $M$ defined by the metric $\eta$.  Thus the fields of the theory in the multisymplectic picture are provided by sections $(A,P)$ of the double bundle $P(E)\to E \to M$.

We will formulate our theory directly in terms of the natural fields $A, P$, and we will write the action functional following the general principle, eq. \eqref{action_phip}:
\begin{equation}\label{ymap}
S_{\mathrm{YM}}(A,P) = \int_M P_a^{\mu\nu} \d A_\mu^a \wedge \d x^{m-1}_\nu - H(A,P) \mathrm{vol}_M \, .
\end{equation}
The Hamiltonian function is defined as,
\begin{equation}\label{hamiltonian}
H(A,P) = \frac{1}{2} \epsilon_{bc}^aP^{\mu\nu}_a A_\mu^bA_\nu^c + \frac{1}{4}P^{\mu\nu}_a P_{\mu\nu}^a \, ,
\end{equation}
where the indexes $\mu\nu$  ($a$) in $P_a^{\mu\nu}$ have been lowered (raised) with the aid of the Lorentzian metric $\eta$ (the Killing-Cartan form on $\mathfrak{g}$, respect.). 
Expanding the right hand side of eq. \eqref{ymap}, we get\footnote{The minus sign in front comes form the expansion of $P_a^{\mu\nu} dA_\mu^a \wedge \d x^{m-1}_\nu$, that gives $P^{\mu\nu}_a (\partial_\nu A_\mu^a - \partial_\mu A_\nu^a) \mathrm{vol}_M$.},
\begin{equation}\label{ymP}
S_{\mathrm{YM}}(A,P) = -\int_M \frac{1}{2} \left[ P^{\mu\nu}_a (\partial_\mu A_\nu^a - \partial_\nu A_\mu^a + \epsilon_{bc}^a A_\mu^b A_\nu^c) +  \frac{1}{2} P^{\mu\nu}_aP^a_{\mu\nu} \right] \, \mathrm{vol}_M \, .
\end{equation}
Notice that if $A$ is given by eq. \eqref{connectionA}, then, its curvature is given by,
\begin{eqnarray}\label{Fmunu}
F_A &=& \d_A A = \d A + \frac{1}{2}[A\wedge A] = F_{\mu\nu} \d x^\mu \wedge \d x^\nu \\
&=& \frac{1}{2}\left( \partial_\mu A_\nu^a - \partial_\nu A_\mu^a + \epsilon_{bc}^a A_\mu^b A_\nu^c\right) \d x^\mu \wedge \d x^\nu \otimes \xi_a\, . \nonumber
\end{eqnarray}
Thus the previous expression for the Yang-Mills action becomes,
$$
S_{\mathrm{YM}}(A,P) = -\int_M \left[ P_a^{\mu\nu} F_{\mu\nu}^a + \frac{1}{4}P^{\mu\nu}_aP^a_{\mu\nu} \right] \, \mathrm{vol}_M \,.
$$
The Euler-Lagrange equations of the theory are very easy to obtain from the previous expression, they are,
\begin{equation}\label{ymfo}
\frac{1}{2}P_{\mu\nu}^a = - F_{\mu\nu}^a \, , \qquad \partial_\mu P_a^{\mu\nu} + \epsilon_{ab}^c A_\mu^b P_c^{\mu\nu}= 0 \, .
\end{equation}

\subsection{The canonical formalism near the boundary}\label{sec:canonical}
In order to obtain an evolution description for Yang-Mills and to prepare the ground for the discussion of its canonical quantization, we need to introduce a local time parameter. 

 In the case that $M$ is a Lorenztian manifold it is customary to assume that $M$ is globally hyperbolic (even if far less strict causality assumptions on $M$ would suffice), therefore the time parameter can be chosen globally.   
Actually we will only assume that a collar $U_{\epsilon} = (-\epsilon, 0]\times \partial M$ around the boundary can be chosen and so that a choice of a time parameter $t = x^0$ can be made near the boundary that would be used to describe the evolution of the system.  The fields of the theory would then be considered as fields defined on a given spatial frame that evolve in time for $t \in (-\epsilon, 0]$.  

The dynamics of such fields would be determined by the restriction of the Yang-Mills action \eqref{ymP} to the space of fields 
on $U_{\epsilon}$,
\begin{equation}\label{ymPepsilon}
S_{\mathrm{YM},U_{\epsilon}}(A,P)= - \int_{-\epsilon}^0 \d t \int_{\partial M} \mathrm{vol}_{\partial M}  \left[ P^{\mu\nu}_a F_{\mu\nu}^a + \frac{1}{4} P^{\mu\nu}_aP^a_{\mu\nu} \right] \, ,
\end{equation}
where now we are assuming that the collar $U_\epsilon$ is strongly hyperbolic and $\mathrm{vol}_{U_\epsilon} = \d t \wedge \mathrm{vol}_{\partial M}$ where $\mathrm{vol}_{\partial M}$ is the canonical volume defined by the restriction of the metric $\eta$ to the boundary.

  Expanding \eqref{ymPepsilon} we obtain,
\begin{eqnarray*}
S_{\mathrm{YM},U_\epsilon} (A,P) & = & -\frac{1}{2} \int^0_{- \epsilon} \d t \int_{\partial M} \mathrm{vol}_{\partial M}
\left[ P^{\mu\nu}_a \left( \partial_\mu A_\nu^a - \partial_\nu A_\mu^a + \epsilon_{bc}^a A_\mu^b A_\nu^c \right) + \frac{1}{2}P^{\mu\nu}_a P_{\mu\nu}^a \right] \\
& = & - \frac{1}{2} \int^0_{- \epsilon} \d t \int_{\partial M} \mathrm{vol}_{\partial M}
\left[  P^{k0}_a \left( \partial_k A_0^a - \partial_0 A^a_k + \epsilon^a_{bc} A^b_k A^c_0 \right) \right. + \\
&+& P^{0k}_a \left(\partial_0 A^a_k  - \partial_k A_0^a + \epsilon^a_{bc} A^b_0 A^c_k \right) +  \\
&+&  \left. P^{kj}_a \left( \partial_k A^a_j - \partial_j A^a_k +  \epsilon^a_{bc}A^b_kA^c_j \right)  +  \frac{1}{2} P^{k0}_a P^a_{k0} + \frac{1}{2} P^{0k}_a P^a_{0k}  +  \frac{1}{2}P^{kj}_aP_{kj}^a \right]  \\
&=& \int^0_{- \epsilon} \d t \int_{\partial M} \mathrm{vol}_{\partial M} 
\left[  P^{k0}_a \left( \partial_0 A^a_k - \partial_k A_0^a - \epsilon^a_{bc} A^b_k A^c_0 \right) \right. + \\
&- & \left. \frac{1}{2} P^{kj}_a \left( \partial_k A^a_j - \partial_j A^a_k +  \epsilon^a_{bc}A^b_kA^c_j \right)  -   \frac{1}{2}P^{k0}_a P^a_{k0}  -  \frac{1}{4}P^{kj}_aP_{kj}^a \right] \, .
\end{eqnarray*}
In the previous expressions $\epsilon_{bc}^a$ denote the structure constants of the Lie algebra $\mathfrak{g}$ with respect to the basis $\xi_a$, that is $[\xi_b, \xi_c] = \epsilon_{bc}^a \xi_a$.
Notice that  $\epsilon^a_{bc}A^b_0A^c_0=0$ because for fixed a, ${\epsilon}^a_{bc}$ is skew-symmetric.  Moreover the indexes $\mu$ and $a$ have been pushed down and up by using the metric $\eta$ and the Killing-Cartan form $\langle\cdot, \cdot \rangle$ respectively.

 In equation \eqref{ymP} we introduced the assumption that $P$ is a bivector, i.e., $P_a^{\mu\nu}$ is skew symmetric in $\mu$ and $\nu$. Therefore $P_a^{00} = 0$, and also $P^{k0}_a P^a_{k0} = P^{0i}_a P^a_{0i}$, because $P^{k0} = - P^{0k}$, etc.  This assumption will be justified later on (see Sect. \ref{sect:Legendre})

The previous expression acquires a clearer structure by introducing the appropriate notations for the fields restricted at the boundary and assuming that they evolve in time $t$.
Thus the pull-backs of the components of the fields $A$ and $P$ to the boundary will be denoted respectively as,
\begin{eqnarray*}
a^a_k &:=& A^a_k\mid_{\partial M} ; \qquad  a = (a^a_k) \, , \qquad a^a_0 := A^a_0\mid_{\partial M} ; \qquad a_0 = (a^k_0) \, , \\
p^k_a &:=& P^{k0}_a \mid_{\partial M} ; \qquad p = (p^k_a)\, , \qquad p^0_a := P^{00}_a\mid_{\partial M}= 0 ; \qquad p_0 = (p^0_a)=0 \, , \\
\beta^{ki}_a &:=& P^{ki}_a\mid_{\partial M} ; \qquad \beta =( \beta^{ki}_a) \, .
\end{eqnarray*}
Given two fields at the boundary, for instance $p$ and $a$, we will denote as usual by $\langle p, a\rangle$ the following expression: 
$$
\langle p, a\rangle = \int_{\partial M} p_a^\mu a_\mu^a \, \mathrm{vol}_{\partial M}\, ,
$$
and the contraction of the inner (Lie algebra) indices by using the Killing-Cartan form and the integration over the boundary is understood. 

Introducing the notations and observations above in the expression for $S_{\mathrm{YM}, U_\epsilon}$ we obtain,
\begin{eqnarray}
S_{\mathrm{YM},U_\epsilon}(A,P) &=& \int^0_{- \epsilon}\d t \int_{ \partial M} \mathrm{vol}_{\partial M} 
\left[ p^k_a \left( \dot{a}^a_k - \partial_ka_0^a -\epsilon^a_{bc} a^b_k a^c_0 \right) \right. + \nonumber
\\
&-& \frac{1}{2}  \beta^{ki}_a \left( \partial_k a^a_i - \partial_i a^a_k + \epsilon^a_{bc}a^b_ka^c_i \right) -
\frac{1}{4}  \beta^{ki}_a \beta^a_{ki}  - \frac{1}{2}p_a^kp^a_k  = \nonumber \\
&=& \int^0_{- \epsilon} \d t \, \mathcal{L}(a,\dot{a},a_0,\dot{a}_0,p,\dot{p}, \beta,\dot{\beta}) \,   \label{ym_boundary}
\end{eqnarray}
where now $\mathcal{L}$ denotes the boundary Lagrangian, Eq. \eqref{densityL}, and depends on the restrictions to the boundary of the fields of the theory.
Collecting terms and simplifying we can then write $\mathcal{L}$ as,
\begin{equation}
 \mathcal{L}(a,\dot{a},a_0,\dot{a}_0,p,\dot{p},\beta,\dot{\beta}) 
 = Ê\langle p,\dot{a} - \d_a a_0 \rangle - \langle \beta , F_a \rangle - \frac{1}{2} \langle p,p \rangle - \frac{1}{4} \langle \beta, \beta \rangle \, .
\end{equation}

Now we can find the\ Euler-Lagrange equations corresponding to the Lagrangian function $\mathcal{L}$ as an infinite-dimensional mechanical system defined on the configuration space $P(E) = \{ a, a_0, p, \beta \}$.  Notice that the fields $a$, $p$ are 1-forms on $\partial M$ with values in the Lie algebra $\mathfrak{g}$, while the field $a_0$ is a function on $\partial M$ with values in $\mathfrak{g}$, and the field $\beta$ is a 2-form on $\partial M$ with values in $\mathfrak{g}$ too.   Thus the configuration space is the space of sections of the bundle $(T^*M\oplus T^*M\oplus \Lambda^2(T^*M)\oplus \mathbb{R} )\otimes \mathfrak{g}$.   

Euler-Lagrange equations will have the form:
$$
\frac{d}{\d t} \frac{\delta \mathcal{L}}{\delta \dot{\chi}} = \frac{\delta \mathcal{L}}{\delta \chi} \, ,
$$
where $\chi \in P(E)$ and $\delta /\delta \chi$ denotes the variational derivative of the functional $\mathcal{L}$.

Thus for $\chi = p$ we obtain,
$$
\frac{\delta \mathcal{L}}{\delta \dot{p}}=0, \quad \mathrm{hence} \quad 0 = \frac{\delta \mathcal{L}}{\delta p} = -p + \dot{a} - \d_a a_0 \, ,
$$
and thus,
 \begin{equation}\label{control}
 \dot{a} = p  +  \d_a a_0 \, .
 \end{equation}
 This equation corresponds to the Legendre transformation of the velocity and agrees with the standard minimal coupling definition of the momenta $p = \dot{a} - \mathrm{d}_a a_0$.
 
For $\chi = \beta$ we obtain,
$$
\frac{\delta \mathcal{L}}{\delta \dot{\beta}} = 0, \quad \mathrm{thus} \quad 0=\frac{\delta \mathcal{L}}{\delta \beta } = - F_a - \frac{1}{2}{\beta}
$$
and consequently,
\begin{equation}\label{betafa}
\beta = -2 F_a \, .
\end{equation}

For $\chi = a$ we obtain,
$$
\frac{\delta \mathcal{L}}{\delta \dot{a}} = p , \quad \mathrm{hence} \quad \dot{p} = \frac{d}{\d t}\frac{\delta \mathcal{L}}{\delta \dot{a}} = \frac{\delta \mathcal{L}}{\delta a} =  \mathrm{d}^*_a\beta + [p,a_0].
$$
Thus we get the equation determining the evolution of the momenta field (the Yang-Mills electric field) $p$:
\begin{equation}\label{pdot}
\dot{p} = \mathrm{d}^*_a\beta + [p,a_0] \, .
\end{equation}

Finally for $\chi = a_0$ we obtain,
$$
\frac{\delta \mathcal{L}}{\delta \dot{a_0}} = 0, \quad \mathrm{and \,\,\, therefore,} \quad \frac{\delta \mathcal{L}}{\delta a_0} = \mathrm{d}_a^*p  \, .
$$
Thus we obtain,
\begin{equation}\label{gauss}
\mathrm{d}^*_ap = 0
\end{equation}
that must be interpreted as Yang-Mills Gauss law (in the absence of charges). 
Thus we have two evolution equations, \eqref{control} and 
\eqref{pdot}, and two constraint equations \eqref{betafa} and
 \eqref{gauss}.

Notice that the field $a_0$ is undetermined.  This fact, clearly a consequence of the gauge invariance of the theory,
will be interpreted in the next section.

We will study the consistency of the previous equations in the following section.

\subsection{The Legendre transform}\label{sect:Legendre}

\subsubsection{The Legendre transform in the bulk}

So far we have presented a covariant Hamiltonian theory, equation following (4.5), whose Euler-Lagrange equations are equivalent to Yang-Mills equations. However it is not automatically true that such theory is equivalent to the standard Yang-Mills theory.  The standard Yang-Mills theory is a Lagrangian theory determined by a Lagrangian density which is nothing but the square norm of the curvature $F_A$ of the connection 1-form $A$, and its action the $L^2$ norm of $F_A$, i.e.
\begin{equation}\label{yms}
S = - \frac{1}{4}\int_{M} \mathrm{Tr\,} (F_A \wedge \star F_A) = \int_M L_{\mathrm{YM}}(A) \mathrm{vol}_M \, .
\end{equation}
Standard quantum field theories describing gauge interactions use exactly this Lagrangian description (and provide accurate results).   Thus if we will assume that the correct Yang-Mills theory is provided by the action above, eq. \eqref{yms}, then we would like to relate the covariant Hamiltonian picture above to this Lagrangian picture.  

For this task we have to introduce the natural extension of Legendre transform to the setting of covariant first order Lagrangian field theories.     The Legendre transform is defined \cite{Ca91} as the bundle map $\mathcal{F}L_{YM} \colon J^1E \to P(E)$, as $\mathcal{F}L_{YM}(x^\mu, A_\mu^a; A_{\mu\nu}^a) = (x^\mu, A_\mu^a; P_a^{\mu\nu})$, where
$$
P_a^{\mu\nu} = \frac{\partial L_{\mathrm{YM}}}{\partial A_{\mu\nu}^a} \, 
$$
and $L_{\mathrm{YM}} = - \frac{1}{4} \mathrm{Tr\,} (F_A \wedge \star F_A)$.
Now recall that $\alpha \wedge \star \beta = (\alpha, \beta )_\eta \mathrm{vol}_M$, $\alpha, \beta$, $k$-forms,  where $(\cdot, \cdot )_\eta$ denotes the inner product on $k$-forms.  Thus we will write $\alpha \wedge \star \beta = \alpha_{\mu_1\cdots \mu_k} \beta^{\mu_1 \cdots \mu_k} \mathrm{vol}_M$ where we have raised the indexes by using the 
$\eta^{\mu\nu}$.   Hence,
\begin{equation}\label{LYM}
L_{YM} = \frac{1}{2} F_{\mu\nu} F^{\mu\nu} \, .
\end{equation}
Hence in bundle coordinates $(x^\mu, A_\mu^a; A_{\mu\nu}^a)$, we have,
\begin{equation}F_{\mu\nu} = \frac{1}{2}\left( A_{\nu\mu}^a - A_{\mu\nu}^a + \epsilon_{bc}^a A_\mu^b A_\nu^c\right) \, .
\end{equation}
Thus 
$$
P_a^{\mu\nu} = F^{\mu\nu}_a \, .
$$
Notice that on the graph of the Legendre map, the Yang-Mills action in the Hamiltonian first order formalism, eq. \eqref{ymP}, is just, up to a coefficient, the previous action eq. \eqref{yms}.

It was mentioned at the end of Section \ref{sec:canonical} that the momenta fields $P^{\mu\nu}$ are skew-symmetric in the indices $\mu$ and $\nu$.  Notice that from the definition of the momenta fields as sections of the bundle $P(E)$ there is no restriction on them.  However because Yang-Mills theories are Lagrangian theories, the Legendre transform selects a subspace on the space of momenta that corresponds to fields $P$ which are skew-symmetric on the indices $\mu$, $\nu$.


\subsection{The presymplectic formalism: Yang-Mills at the boundary and reduction}
  As discussed in general in section $3.2$, we define the extended Hamiltonian, $\mathcal{H}$, so that $\mathcal{L} = \langle p, \dot{a}\rangle - \mathcal{H}$.
 Thus  
\begin{equation}\label{PH}
\mathcal{H} (a,\beta) =  \langle p, \d_a a_0 \rangle + \frac12 \langle p, p \rangle +  \langle \beta, F_a + \frac{1}{2} \beta \rangle \, .
\end{equation}

Thus the Euler-Lagrange equations can be rewritten as 
\begin{equation}\label{PMP1}
\dot{a} = \frac{\delta \mathcal{H}}{\delta p}; \quad \dot{p} = - \frac{\delta \mathcal{H}}{\delta a} \, ,
\end{equation}
\begin{equation}\label{PMP0}
\frac{\delta \mathcal{H}}{\delta a_0} = 0 \,
\end{equation}

\begin{equation}\label{PMP2}
\frac{\delta \mathcal{H}}{\delta \beta} = 0 \, .
\end{equation}

We denote again by $\varrho \colon \mathcal{M} \to T^*\mathcal{F}_{\partial M}$ the canonical projection $\varrho(a,a_0,p,\beta)=(a,a_0,p)$.  Let $\omega_{\partial M}$ denote the form on the cotangent bundle $T^*\mathcal{F}_{\partial M}$, 
$$
\omega_{\partial M} =  \delta a \wedge \delta p . 
$$
We will denote again by $\Omega$ the pull-back of this form to $\mathcal{M}$ along $\varrho$, i.e., $\Omega = \varrho^*\omega_{\partial M}$.  Clearly, $\ker {\Omega} = \mathrm{span} \{ \delta /\delta \beta, \delta/\delta a_0 \}$, and we have the particular form that Thm.  \ref{presymplectic_equation} takes here.

\begin{theorem}
The solution to the equation of motion defined by  the Lagrangian $L_{\mathrm{YM}}$, i.e. the Yang-Mills equations,
are in one-to-one correspondence with the integral 
curves of the presymplectic system $(\mathcal{M},\Omega,\mathcal{H})$, i.e. with the integral curves of the vector field $\Gamma$ on $\mathcal{M}$ such that $i_\Gamma\Omega= \mathrm{d} \mathcal{H}$.
\end{theorem}          

The primary constraint submanifold $\mathcal{M}_1$ is
 defined by the two constraint equations,
$$
\mathcal{M}_1= \{(a,a_0,p,\beta)|\mathcal{F}_a + \beta = 0, d_a^*p = 0 \} \, .
$$

Since $\beta$ is just a function of $a$, we have that $\mathcal {M}_1\cong \{(a,a_0,p)| d_a^*p = 0\}$ and 
$\ker \Omega|_{\mathcal{M}_1} = \mathrm{span} \{\frac{\partial}{\partial a_0}\}$.    Thus
$\mathcal{M}_2 = \mathcal{M}_1/(\ker \Omega|\mathcal{M}_1) \cong \{(a,p)| \mathrm{d}_a^*p = 0\}.$


\subsection{Gauge transformations: symmetry and reduction}

The group of gauge transformations $\mathcal{G}$, i.e, the group of automorphisms of the principal bundle $P$ over the identity, is a fundamental symmetry of the theory.      Notice that the action $S_{\mathrm{YM}}$ is invariant under the action of $\mathcal{G}$ (however it is not true that $H$ is $\mathcal{G}$-invariant).

The quotient of the group of gauge transformations by the normal subgroup of identity gauge transformations at the boundary defines the group of gauge transformations at the boundary $\mathcal{G}_{\partial M}$, and it constitutes a symmetry group of the theory at the boundary, i.e. it is a symmetry group both of the boundary Lagrangian $\mathcal{L}$ and of the presymplectic system $(\mathcal{M}, \Omega, \mathcal{H})$.    We may take advantage of this symmetry to provide an alternative description of the constraints found in the previous section.    

\begin{proposition}
With the notations above, $\mathcal{J}(a,p) = \d_a^*p$.
\end{proposition} 

\begin{proof}  The moment map $\mathcal{J} \colon T^*\mathcal{F}_{\partial M} \to \mathfrak{g}_{\partial M}^*$ is given by,
$$
\langle \mathcal{J} (a,p), \xi \rangle = \langle p, \xi_{\mathcal{F}_{\partial M}}\rangle = \langle
 p, \d_a\xi \rangle \, ,
$$
because the gauge transformation $g_s = \exp s \xi$ acts in $a$ as $a \mapsto g_s\cdot a =   g_s^{-1} a g_s + g_s^{-1}\d g_s$
  and
the induced tangent vector is given by,
$$
\xi_{\mathcal{A}_{\partial M}} (a)= \frac{\d}{\d s}  g_s \cdot a \mid_{s = 0} = \d_a\xi \, .
$$

\end{proof}

Let $\mathcal{A}_{\partial M}$ denote the space of connections $a$ defined on the boundary $\partial M $.  The constraint submanifold 
 $\mathcal{M}_1$ projected to the space $T^*\mathcal{A}_{\partial M}$, by means of the projection map $(a,a_0,p) \mapsto (a,p)$, is such that $\mathcal{C} = \mathcal{J}^{-1}(\mathbf{0})$. This is exactly the situation depicted in Sect.\ref{reduction_gauge}. Hence the standard Marsden-Weinstein reduction, eq. \eqref{MWR}, will give the reduced phase space,
$$
\mathcal{R}_{\mathrm{YM}} = \mathcal{J}^{-1}(\mathbf{0})/ \mathcal{G}_{\partial M} \, .
$$
and its Hamiltonian,
$$
h ([a], [p] )= \frac{1}{2}\langle p, p \rangle - \frac{1}{2}\langle F_a, F_a\rangle \, ,
$$
where $[a]$ and $[p]$ denote equivalence classes of connections and momenta with respect to the action of the gauge group $\mathcal{G}_{\partial M}$.   Notice that both terms in the Hamiltonian function $h$ are $\mathcal{G}_{\partial M}$-invariant, and the Hamiltonian system $h$ defined on the reduced phase space $\mathcal{R}_{\mathrm{YM}}$ 
has the structure of an infinite-dimensional mechanical system with potential function $V([a]) = \frac{1}{2}|| F_a ||^2$.

The reduction of the boundary values of solutions of Yang-Mills equation in the bulk is of course, an isotropic submanifold of the reduced space.  In the case where $M$ is Riemannian, an existence and uniqueness theorem for solutions of Yang-Mills equations on manifolds with boundary can be proved and hence this submanifold, following the proof of Theorem 2.7, is a Lagrangian submanifold.

 \section{Conclusions and discussion}

It has been shown that the multisymplectic geometry of the covariant phase space $P(E)$ provides a convenient framework to study first order covariant Hamiltonian field theories on manifolds with boundaries.  In particular it induces a natural presymplectic structure on the total space of fields at the boundary whose reduction provides the symplectic phase space of the theory.  The solution of the Euler-Lagrange equations on the bulk induce an isotropic  submanifold in the reduced symplectic phase space at the boundary. Provided that the boundary conditions are well-posed, this submanifold is in fact Lagrangian.

The gauge symmetries of the theory fit nicely into the picture and the symplectic reduction of the theory at the boundary induced by the moment map, i.e., by the conserved charges of the theory, is in perfect agreement with the presymplectic analysis of the theory.   Various instances are discussed illustrating the main features of the theoretical framework:  the real scalar field, the Poisson $\sigma$-model and Yang-Mills theories.  Each of them allows as to stress different aspects of the theory.  The regular situation for the scalar field, the coisotropic structure at the boundary in the case of the Poisson $\sigma$-model and the reduction using the moment map at the boundary in the case of Yang-Mills theories.   

The theory presented in this work is particularly well suited for describing Palatini's gravity. C. Rovelli's
 \cite{Ro04}, \cite{Ro06}, can be read in part as seeking and arguing for precisely such a theory. We interpret Rovelli's canonical form ${\Theta}_H$ as alluding to a multisymplectic structure in the bulk. Such aspects will be discussed in a subsequent paper where the reduction of Topological Field Theories at the boundary and Palatini's gravity will be discussed from a common perspective.
 

\section*{Acknowledgements}\label{sec:acknowledgements}
A.I. was partially supported by the Community of Madrid project QUITEMAD+, S2013/ICE-2801, and MINECO grant MTM2014-54692-P.   Part of this work was completed while A.S. was a guest of the Mathematics department at University Carlos III de Madrid, and supported by Spain's Ministry of Science and Innovation Grant MTM2010-21186-C02-02. 
She thanks the Mathematics department for their warm hospitality and for their financial support.
A.S. also thanks Nicolai Reshetikhin for suggesting to her a problem that motivated this work. The solution to that other problem will appear in another publication. Finally, the authors would like to thank the referee for carefully reading over the manuscript and for drawing their attention to important points that were in need of clarification.


\end{document}